\documentclass[%
 reprint,
 amsmath,amssymb,
 aps,
prb,
]{revtex4-1}
 \usepackage{amsmath,amsfonts,amssymb,graphicx,multirow,overpic,dsfont,bm,braket}
 \usepackage{bbold}
\usepackage{graphicx}% Include figure files
\usepackage{dcolumn}% Align table columns on decimal point
\usepackage{bm}% bold math
\usepackage{amsmath}
\usepackage{amssymb}
\usepackage{amsthm}
\newtheorem{theorem}{Theorem}[section]

\newtheorem{lemma}[theorem]{Lemma}

\newtheorem{defn}{Definition}[section]
\renewcommand{\vec}[1]{\mathbf{#1}}
\newcommand{\matr}[1]{\bm{#1}}     % ISO complying version
\newcommand{\dv}[1]{\deg_v{#1}}
\newcommand{\du}[1]{\deg_u{#1}}

\begin{document}

\preprint{APS/123-QED}

\title{Classifying local fractal subsystem symmetry protected topological phases}% Force line breaks with \\
% \thanks{A footnote to the article title}%

\author{Trithep Devakul}
%  \altaffiliation[Also at ]{Physics Department, XYZ University.}%Lines break automatically or can be forced with \\
% \author{Second Author}%
%  \email{Second.Author@institution.edu}
\affiliation{%
Department of Physics, Princeton University, Princeton, NJ 08540, USA
}%
\affiliation{%
Kavli Institute for Theoretical Physics, University of California, Santa Barbara, CA 93106, USA
}%

\date{\today}% It is always \today, today,
             %  but any date may be explicitly specified

\begin{abstract}
We study symmetry-protected topological (SPT) phases of matter in $2$D protected by symmetries acting on fractal subsystems of a certain type.
Despite the total symmetry group of such systems being subextensively large, we show that only a small number of phases are actually realizable by local Hamiltonians.
Which phases are possible depends crucially on the spatial structure of the symmetries, and we show that in many cases no non-trivial SPT phases are possible at all.
In cases where non-trivial SPT phases do exist, we give an exhaustive enumeration of them in terms of their locality.  
\end{abstract}

\maketitle

\section{Introduction}
Understanding and classifying the possible phases of matter has been a long running goal of condensed matter physics.
In systems without any symmetries, one can have topological ordered phases which are long range entangled.
With symmetries present, there are many more possibilities: the symmetry may be spontaneously broken, it may enrich an existing topological order, or it may lead to non-trivial short range entangled phases called symmetry-protected topological (SPT) phases~\cite{Chen2011-et,Chen2011-ss,Schuch2011-jx,Pollmann2010-fl,Senthil2015-tp,Chen2013-gq,Chen2012-tt}.

Recently, a new type of symmetry, called ``subsystem symmetries'', has been gaining interest for a number of reasons.
These are symmetries which act on only a rigid (subextensive) subsystem of the full system, for example, along only a row or a column of a square lattice.
Systems with such symmetries show up in a variety of contexts~\cite{Batista2005-yr,Nussinov2009-ld,Nussinov2009-sw,Xu2004-oj,Xu2005-df,Johnston2012-rt,Vijay2016-dr, PhysRevB.81.184303}.
Note that there is a distinction between subsystem symmetries and higher-form symmetries~\cite{Gaiotto2015-cj}, which act on deformable manifolds.
One reason for the recent interest is due to their connection to fracton topological order~\cite{Vijay2016-dr,Chamon2005-fc,Haah2011-ny,Bravyi2011-fl,Yoshida2013-of,Vijay2015-jj,Pretko2017-ha,Pretko2017-ej,Prem2018-nv,Ma2017-cb,He2017-eq,Gromov2017-zm,Nandkishore2018-ee}.
Namely, systems in $D=3$ dimensions with subsystem symmetries of along $d=2$ planes exhibit a gauge duality to (type-I) fracton topological ordered phases~\cite{Vijay2016-dr,Williamson2016-lv,Shirley2018-en,You2018-as}.
More generally, this can be extended to systems with dimensions $D\geq 3$ and symmetries along regular $1<d<D$ subsystems, whose gauge dual exhibits a generalized fracton topological order.
The case $d=D$ is simply the duality of a model with some global symmetry and a (non-fracton) topologically ordered state, e.g. the gauge dual of the $\mathbb{Z}_2$ symmetric Ising model in $D\geq 2$ is a $\mathbb{Z}_2$ topological order.
The case where $d=1$ is another extreme case, whose gauge dual does not correspond to a topological order.
These should be thought of in analogy to the $D=d=1$ Ising chain, which is dual to another Ising model under the gauge duality.
In the presence of a symmetry group $G$, it is now well known that bosonic $D=d=1$ chains may be classified according to the second cohomology group $\mathcal{H}^2[G,U(1)]$, and may be understood in terms of how the symmetry acts as a projective representation on the edges or under symmetry twists~\cite{Chen2011-ss,Else2014-ar,Levin2012-dv,Wen2017-ak,Wen2014-iz,Barkeshli2014-zg,Tarantino2016-vd,Zaletel2014-ny}.
% Nevertheless, there are many things that we can learn from studying this simplest case.
% It is now well known that bosonic $D=d=1$ chains in the presence of a symmetry group $G$ may be classified according to the second cohomology group $\mathcal{H}^2[G,U(1)]$, which may be understood in terms of how the symmetry acts as a projective representation on the edges or under symmetry twists~\cite{Chen2011-ss,Else2014-ar,Levin2012-dv,Wen2017-ak,Wen2014-iz,Barkeshli2014-zg,Tarantino2016-vd,Zaletel2014-ny}.

Going to one higher dimension, $D=2$, $d=1$, we have two dimensional systems with symmetries acting along rigid lines.
It was recently appreciated that such symmetries could protect non-trivial SPT phases, called subsystem SPT phases~\cite{You2018-em}.
An example of such a phase is the 2D cluster state on the square lattice~\cite{Raussendorf2001-xm}, where it is shown that any state within this subsystem SPT phase is useful as a resource for universal measurement based quantum computing (MBQC)~\cite{Else2012-li,Raussendorf2018-nh}, providing a generalization of the connection between MBQC and SPT phases from one dimension~\cite{Else2012-ie,Else2012-li,Miller2015-gl,Stephen2017-cp,Raussendorf2017-gb}.
A classification of such subsystem SPT phases was realized recently in Ref~\onlinecite{Devakul2018-dv} by the present author and colleagues, and relied on the definition of a modified (weaker) equivalence relation between phases.
The reason this was needed in this case is due to the existence of ``subsystem phases'': cases where two states which differ along only a subsystem may belong to distinct phases of matter.
For instance, consider a $D=2$ trivial symmetric state, but along some of the ($d=1$) subsystems, we place a 1D SPT (in such a way that all symmetries are still respected).
This, now, as a whole represents a non-trivial 2D phase of matter protected by the subsystem symmetries, despite looking trivial in most of the bulk.
Furthermore, the existence of such phases means that in the thermodynamic limit where system size is taken to infinity, there are an infinite number of subsystems, and so an infinite number of possible phases.  
The problem with this is that it now takes a subextensive (growing as $\mathcal{O}(L)$ in local systems of size $L\times L$) amount of information to convey exactly what phase a system is in, without assuming any form of translation invariance.
In Ref~\onlinecite{Devakul2018-dv}, it was shown that there existed some intrinsic global ``data'', which we call $\beta$, which is insensitive to the presence of subsystem phases.
All the infinite phases of such a system could therefore be grouped into equivalence classes and classified according to $\beta$.  
This classification has the nice interpretation of being a classification of phases \emph{modulo} lower-dimensional SPT phases, and is related to the problem of classifying $3$D (type-I) fracton topological orders modulo $2$D topological orders~\cite{Shirley2017-fi,Shirley2018-bx,Shirley2018-jy,Shirley2018-yj,Shirley2018-en}.
There is also a connection between this classification and the appearance of a spurious topological entanglement entropy~\cite{Williamson2018-td,Devakul2018-dv,PhysRevB.94.075151,KitaevPreskill,levinwenentanglement}.
The key idea is that a new tool, in this case the modified phase equivalence relation, was necessary in the classification of these subsystem SPT phases.

The topic of interest in this paper is another type  subsystem symmetry: fractal subsystem symmetries.
In $2$D, these may be thought of as ``in-between'' $d=1$ and $d=2$, as symmetries act on subsystems with fractal dimensions $1<d_f<2$.
An early example of such a system is the Newman-Moore model~\cite{Newman1999-fq}, and such models have been useful as a translation invariant toy model of glassiness~\cite{doi:10.1080/14786435.2011.609152} or for their information storage capacity~\cite{Yoshida2011-lk}.
Fractal symmetries have also recently been shown to be able to protect non-trivial SPT phases~\cite{Devakul2018-ru,Kubica2018-dp}.
An example of this is the cluster state on the honeycomb lattice, which (like the square lattice example) has been shown to be useful for MBQC anywhere in the SPT phase~\cite{Devakul2018-di,Stephen2018-qn}.
Here, we wish to ask the more general question of what SPT phases are even possible in such systems with fractal symmetries.
Note that in higher dimensions ($D\geq 3$), similar to models with regular $d$-dimensional subsystem symmetries, the gauge dual of a fractal symmetric model may also result in (type-II) fracton topological order~\cite{Williamson2016-lv,Yoshida2013-of,Devakul2018-ru}, for which very little is currently known about their classification.

Our main finding is that systems with fractal subsystem symmetries are free from subsystem phases and the associated problems that existed for line-like $d=1$ subsystem SPTs.
The key factor at play here is \emph{locality}.
Although the total number of phases is still infinite (a result of the total symmetry group being infinitely large), the vast majority of these phases are highly non-local and therefore unphysical.
If we fix a degree of locality (what we mean by this will be explained) then the number of allowed phases remains finite in the thermodynamic limit.
This allows for the classification of phases directly, without needing to define equivalence classes of phases like before (essentially due to the lack of any ``weak'' subsystem SPT phases~\cite{You2018-em,Devakul2018-dv}).
% The reason for this is that there is a sense in which locality enforces a kind of periodicity in the twist phases characterizing the phase, as we will show.

We first begin by reviewing some necessary preliminary topics in Sec~\ref{sec:prelims}.
We then define fractal symmetries in Sec~\ref{sec:fracsyms}, and discuss the possible local SPT phases in Sec~\ref{sec:localphases}.
In Sec~\ref{sec:construct} we give a explicit constructions for local models realizing an arbitrary local SPT phase.
Sec~\ref{sec:pseudosym} deals with irreversible fractal symmetries and introduces the concept of pseudo-symmetries and pseudo-SPTs.
A summary and discussion of the results is presented in Sec~\ref{sec:discussion}.
Finally, a technical proof of the main result is given in Sec~\ref{sec:proof}.

\section{Preliminaries}\label{sec:prelims}
\subsection{Linear Cellular Automata}\label{sec:lca}
We first describe a class of fractal structures which determine the spatial structure of all our symmetries in this work (see Ref~\onlinecite{Yoshida2013-of} for a nice introduction to such fractals and their polynomial representation).
These fractal structures, which are embedded on to a 2D lattice, are generated by the space-time evolution of a 1D cellular automaton (CA).
In particular, the update rule for this 1D cellular automaton will be linear, translation invariant, local, and reversible.  
These terms will all be explained shortly.

Let $a_{i}^{(j)}\in\mathbb{F}_p$ denote the state of the cell at spatial index $i$ at time index $j$.
Each $a_{i}^{(j)}$ can take on values $0,\dots, p-1$ for some prime $p$ ($p=2$ in the cases with Ising degrees of freedom).
We take periodic boundary conditions in $i$ such that $0 \leq i < L_x$, and define $a_{i+L_x}^{(j)}\equiv a_{i}^{(j)}$.
The state of the full cellular automaton at a time $j$ is given by the vector $\vec{a}^{(j)}\in\mathbb{F}_p^{L_x}$ with elements
 $(\vec{a}^{(j)})_i = a_{i}^{(j)}$,
We will use the notation $v_i$ to denote the $i$th element of a vector $\vec{v}$.
Bold lowercase letters will denote vectors, while bold uppercase letters will denote matrices.

The key ingredient of the cellular automaton is its update rule: given the state $\vec{a}^{(j)}$ at time $j$, how is the state $\vec{a}^{(j+1)}$ at the next time step calculated?
We will consider only the family of update rules of the form
\begin{equation}
    a^{(j+1)}_i = \sum_{k=k_a}^{k_b} c_{k} \vec{a}^{(j)}_{i-k}
    \label{eq:lcaevol}
\end{equation}
where $c_k\in\mathbb{F}_p$ is a set of coefficients only non-zero for $k_a\leq k \leq k_b$.
Note that all addition and multiplication is modulo $p$, following the algebraic structure of $\mathbb{F}_p$.
Linearity refers to the fact that each $a^{(j+1)}_i$ is determined by a linear sum of $a^{(j)}_i$.
Thus, we may represent Eq~\ref{eq:lcaevol} as
\begin{equation}
    \vec{a}^{(j+1)} = \matr{F} \vec{a}^{(j)}
\end{equation}
where $\matr{F}\in\mathbb{F}_p^{L_x\times L_x}$ is an $L_x\times L_x$ matrix with elements given by $F_{i^\prime i} = c_{i^\prime-i}$.
For a given initial state $\vec{a}^{(0)}$, the state at any time $j\geq0$ is simply given by $\vec{a}^{(j)}=\mathbf{F}^j \vec{a}^{(0)}$.

Translation invariance refer to the fact that the update rules do not depend on the location $i$, only on the relative location: $\matr{F}_{i^\prime i} = \matr{F}_{i^\prime + n, i+n}$.
Locality means that $\matr{F}_{i^\prime i}$ is only non-zero for small $|i^\prime-i|$ of order $1$.  
In our case, this means that $|k_a|$ and $|k_b|$ should be small $\mathcal{O}(1)$ values.
Finally, reversibility means that only one $\vec{a}^{(j)}$ can give rise to a $\vec{a}^{(j+1)}$.  
In other words, the kernel of the linear map induced by $\matr{F}$ is empty, and one can define an inverse $\matr{F}^{-1}$ (which will generically be highly non-local) such that $\matr{F}^{-1}\matr{F} = \matr{F}\matr{F}^{-1}=\mathbb{1}$.
This is a rather special property which will depend on the particular update rule as well as choice of $L_x$.

While we assume reversibility for much of this paper, we note that fractal SPTs exist even when the underlying CA is irreversible.
We call such phases pseudo-SPT phases, and are discussed in Sec~\ref{sec:pseudosym}.

\subsection{Polynomials over finite fields}

Cellular automata with these update rules may also be represented elegantly in terms of polynomials with coefficients in $\mathbb{F}_p$.
By this we mean polynomials $q(x)$ over a dummy variable $x$ of the form
\begin{equation}
    q(x) = \sum_{i=0}^{\delta_q} q_i x^i
\end{equation}
where each $q_i\in\mathbb{F}_p$, and the degree $\delta_q\equiv \deg q(x)$ is finite.
The space of all such polynomials is denoted by the polynomial ring $\mathbb{F}_p[x]$.
A state $\vec{a}^{(j)}$ of the cellular automaton may be described by such a polynomial, $a^{(j)}(x)$,
\begin{equation}
    a^{(j)}(x) = \sum_{i=0}^{L_x-1} a^{(j)}_i x^i .
\end{equation}
In the case of periodic boundary conditions one should also work with the identity $x^{L_x}=1$.

Application of the update rule is expressed most simply in the language of polynomials.
Let us define $f(x)$ to be a Laurent polynomial, i.e. $f(x) = \tilde{f}(x) x^{k_a}$ where $\tilde{f}(x)\in\mathbb{F}_p[x]$ is a polynomial (and $k_a$ may be negative), given by
\begin{equation}
    f(x) = \sum_{k=k_a}^{k_b} c_k x^k
\end{equation}
after which the update rule may be expressed simply as multiplication
\begin{equation}
    a^{(j+1)}(x) = f(x) a^{(j)}(x)
\end{equation}
Given an initial state $a^{(0)}(x)$ then, the state at any future time is simply given by $a^{(j)}(x) = f(x)^j a^{(0)}(x)$.
We will assume $c_{k_a}$ and $c_{k_b}$ are non-zero, and $k_b\neq k_a$ (so that $f(x)$ is not a monomial).

The key property of such polynomials that guarantees fractal structures is that for $q(x)\in\mathbb{F}_p[x]$, one has that
\begin{equation}
    q(x)^{p^n} = q(x^{p^n})
\end{equation}
also known as the ``freshman's dream''.
Suppose we start off with the initial state $a^{(0)}(x)=1$.
After some possibly large time $p^n$, the state has evolved to
\begin{equation}
    a^{(p^n)}(x) =f(x)^{p^n} = f(x^{p^n}) =\sum_{k=k_a}^{k_b} c_k x^{k p^n}
\end{equation}
which is simply the initial state at positions separated by distances $p^{n}$.
At time $p^{n+1}$, this repeats but at an even larger scale.
Thus, the space-time trajectory, $a^{(j)}_i$, of this cellular automaton always gives rise to self-similar fractal structures.

There are various other useful properties that will be used in the proof of Sec~\ref{sec:proof}, one of which is that any polynomial $q(x)\in \mathbb{F}_p[x]$ (without periodic boundary conditions) may be uniquely factorized up to constant factors as 
\begin{equation}
    q(x) = q_1(x) q_2(x) \dots q_n(x)
\end{equation}
where each $q_i(x)$ is an irreducible polynomial of positive degree.  
A polynomial is irreducible if it cannot be written as a product of two polynomials of positive degree.
This may be thought of as a ``prime factorization'' for polynomials.  
% The factorization of $f(x)$ will affect the classification of fractal SPT phases in Sec~\ref{subsec:localtwo}.

\subsection{Projective Representations}
The final topic which should be introduced are projective representation of finite abelian groups.
Bosonic SPTs in 1D are classified by the projective representations of their symmetry group on the edge~\cite{Chen2011-ss,Else2014-ar}.
Similarly, subsystem SPTs for which the subsystems terminate locally on the edges (i.e. line-like subsystems) may also be described by projective representations of a subextensively large group on the edge~\cite{You2018-em,Devakul2018-dv}.
The same is true for fractal subsystem symmetries~\cite{Devakul2018-ru}.

Let $G$ by a finite abelian group.
A non-projective (also called linear) representation of $G$ is a set of matrices $V(g)$ for $g\in G$ that realize the group structure: $V(g_1) V(g_2) = V(g_1 g_2)$ for all $g_1,g_2\in G$.
A projective representation is one such that this is only satisfied up to a phase factor,
\begin{equation}
    V(g_1) V(g_2) = \omega (g_1, g_2) V(g_1 g_2)
\end{equation}
where $\omega(g_1, g_2) \in U(1)$ is called the factor system of the projective representation, and must satisfies the properties
\begin{align}
\begin{split}
\omega (g_1, g_2)\omega (g_1g_2, g_3) &= \omega (g_1, g_2g_3)\omega (g_2, g_3) \\
\omega (1, g_1) = \omega (g_1, 1) &= 1
\end{split}
\end{align}
for all $g_1,g_2,g_3\in G$.
A different choice of $U(1)$ prefactors, $V^\prime(g) = \alpha(g) V(g)$ leads to the factor system
\begin{equation}
    \omega^\prime(g_1,g_2) = \frac{\alpha(g_1 g_2)}{\alpha(g_1)\alpha(g_2)} \omega(g_1,g_2).
    \label{eq:prepgauge}
\end{equation}
for $V^\prime(g)$.
Two factor systems related in such a way are said to be equivalent, and belong to the same equivalence class $\omega$.

Suppose we have a factor system $\omega_1(g_1,g_2)$ of equivalence class $\omega_1$, and a factor system $\omega_2(g_1,g_2)$ of class $\omega_2$.  
A new factor system can be obtained as $\omega(g_1,g_2) = \omega_1(g_1,g_2)\omega_2(g_1,g_2)$, which is of class $\omega\equiv \omega_1 \omega_2$.  
This gives them a group structure: equivalence classes are in one-to-one correspondence with elements of the second cohomology group $\mathcal{H}^2[G,U(1)]$, and exhibit the group structure under multiplication.  

In the case of finite abelian groups, a much simpler picture may be obtained in terms of the quantities
\begin{equation}
    \Omega(g_1,g_2) \equiv \frac{\omega(g_1,g_2)}{\omega(g_2,g_1)}
\end{equation}
which is explicitly invariant under the transformations of Eq~\ref{eq:prepgauge}.
They have a nice interpretation of being the commutative phases of the projective representation
\begin{equation}
    V(g_1) V(g_2) = \Omega(g_1,g_2) V(g_2) V(g_1).
\end{equation}
$\Omega(g_1,g_2)$ has the properties of bilinearity and skew-symmetry in the sense that
\begin{align}
    \Omega(g_1g_2,g_3) &= \Omega(g_1,g_3)\Omega(g_2,g_3)\\
    \Omega(g_1,g_2g_3) &= \Omega(g_1,g_2)\Omega(g_1,g_3)\\
    \Omega(g_1,g_2) &= \Omega(g_2,g_1)^{-1}
\end{align}
These properties mean that  $\Omega(g_1,g_2)$ is completely determined by its value on all pairs of generators of $G$.
Suppose $a_1,a_2\in G$ are two independent generators with orders $n_1, n_2$, respectively. 
Then, one can show that $\Omega(a_1,a_2)^{n_1} = \Omega(a_1,a_2)^{n_2} = 1$, and so $\Omega(a_1,a_2) = e^{2\pi i w/\gcd(a_1,a_2)}$ for integer $w$.
The value of $w$ for every pair of generators provides a complete description of the projective representation, and each of them may be chosen independently.

By the fundamental theorem of finite abelian groups, $G$ may be written as a direct product 
\begin{equation}
G = \mathbb{Z}_{n_1} \otimes \mathbb{Z}_{n_2} \otimes \cdots  \otimes \mathbb{Z}_{n_N}
\end{equation}
where each $n_i$ are prime powers.
Let $a_i$ be the generator of the $i$th direct product of $G$ with order $n_i$, and define $m_{ij}$ through $\Omega(a_i, a_j) = e^{2\pi i m_{ij}/\gcd(n_i,n_j)}$.
Each choice of $0\leq m_{ij}< \gcd(n_i,n_j)$ for $i<j$ corresponds to a distinct projective representation.
Indeed, applying the Kunneth formula, one can compute the second cohomology group
\begin{align}
    \mathcal{H}^2[G,U(1)] = \prod_{i < j} \mathbb{Z}_{\gcd(n_i, n_j)}
\end{align}
There is therefore a one-to-one correspondence between choices of $\{m_{ij}\}$ and elements of $\mathcal{H}^2[G,U(1)]$.

Hence, we may simply refer to the commutative phases $\Omega(g_1,g_2)$ of the generators, $\{m_{ij}\}$, as a proxy for the whole projective representation.

\subsection{1D SPTs and twist phases}\label{sec:1dtwist}
Let us now connect our discussion of projective representations to the classification of 1D SPT phases.
There are various ways this connection can be made, for instance, by looking at edges or matrix product state representations~\cite{Else2014-ar,Schuch2011-jx}.
Here, we will be using symmetry twists~\cite{Chen2011-ss,Else2014-ar,Levin2012-dv,Wen2017-ak,Wen2014-iz,Barkeshli2014-zg,Tarantino2016-vd,Zaletel2014-ny}, which turn out to be a natural probe in the case of 2D fractal symmetries~\cite{Devakul2018-ru}.

Suppose we have a 1D SPT described by the unique ground state of the local Hamiltonian $H$ and
global on-site symmetry group $G$.
Let us take the chain to be of length $L_x$ (taken to be large) with periodic boundary conditions.
The symmetry acts on the system as
\begin{equation}
    S(g) = \prod_{i=0}^{L_x-1} u_i(g)
\end{equation}
for $g\in G$, where $u_i(g)$ is the on-site unitary linear representation of the symmetry element $g$ on site $i$, and $[H, S(g)]=0$.
A local Hamiltonian may always be written as
\begin{equation}
    H = \sum_{i=0}^{L_x-1} H_i
\end{equation}
where the sum is over local terms $H_i$ with support only within some $\mathcal{O}(1)$ distance of $i$.  

The twisting procedure begins by constructing a new Hamiltonian, $H_\mathrm{twist}(g)$, for a given $g\in G$.
We pick a cut across which to apply the twist, $x_\mathrm{cut}$, which can be arbitrary.
Then, define the truncated symmetry operator 
\begin{equation}
S_{\geq}(g) = \prod_{i=x_\mathrm{cut}}^{x_\mathrm{cut}+R} u_i(g) 
\end{equation}
for some $1\ll R \ll L_x$.
The twisted Hamiltonian is given by
\begin{equation}
    H_\mathrm{twist}(g) = \sum_{i=0}^{L_x-1} 
    \begin{cases}
    S_{\geq}(g) H_i S_{\geq}(g)^\dagger & \text{if $H_i$ crosses $x_\mathrm{cut}$}\\
    H_i & \text{else}
    \end{cases}
\end{equation}
thus, the Hamiltonian is modified for $H_i$ near $x_\mathrm{cut}$, but remains the same elsewhere.

We can now define the \emph{twist phase}
\begin{equation}
T(g_1, g_2) = \frac{\langle S(g_1) \rangle_{H_\mathrm{twist}(g_2)}}
{\langle S(g_1) \rangle _{H}}
\end{equation}
which is a pure phase representing the charge of the symmetry $g_1$ in the ground state of the $g_2$ twisted Hamiltonian, relative to in the untwisted Hamiltonian.
Here, $\langle O \rangle_{H}$ means that expectation value of the operator $O$ in the ground state of the Hamiltonian $H$.
It is straightforward to show that $T(g_1,g_2)$ does not depend on where we place the cut, $x_\mathrm{cut}$ (this fact will be used to our advantage when twisting fractal symmetries).
The set of twist phases $T(g_1,g_2)$ is a complete characterization of the state.  
Indeed, the correspondence of the twist phases to the projective representation characterizing a phase can be made by simply
\begin{equation}
    \Omega(g_1,g_2) = T(g_1,g_2).
\end{equation}
as such, we refer to $\Omega(g_1,g_2)$ itself as the twist phases.

An alternate, but equivalent, view is to examine the action of $S_\geq(g_2)$ on the ground state $\ket{\psi}$.  
The action of $S_\geq(g_2)$ on $\ket{\psi}$ must act as identity on the majority of the system, except near $x_\mathrm{cut}$ and $x_\mathrm{cut}+R$, where it may act as some unitary operation,
\begin{equation}
    S_\geq(g_2)\ket{\psi} = U_{g_2} \widetilde{U}_{g_2} \ket{\psi}.
\end{equation}
where $U_{g_2}$ acts near $x_\mathrm{cut}$, and $\widetilde{U}_{g_2}$ acts near $x_\mathrm{cut}+M$.
Then, the twisted Hamiltonian acting on the ground state can be thought of as
\begin{equation}
    H_\mathrm{twist}(g_2)\ket{\psi} = U_{g_2} H U_{g_2}^{\dagger}\ket{\psi}
\end{equation}
 such that the ground state of $H_\mathrm{twist}(g_2)$ is given by $U_{g_2}\ket{\psi}$.
The twist phase is then given by 
\begin{align}
\begin{split}
   \Omega(g_1,g_2) &= \frac{ \bra{\psi} U_{g_2}^\dagger S(g_1) U_{g_2} \ket{\psi}}
   {\bra{\psi} S(g_1) \ket{\psi}}\\
    &=  \bra{\psi} S(g_1)^\dagger U_{g_2}^\dagger S(g_1) U_{g_2} \ket{\psi}
    \end{split}
    \label{eq:meascharge}
\end{align}
which measures the charge of the excitation created by $U_{g_2}$ under the symmetry $S(g_1)$.
Thus, all information regarding the phase is contained within this local unitary matrix $U_{g_2}$ that appears due to a truncated symmetry operator.
% By virtue of the fact that $\ket{\psi}$ is the unique gapped ground state, $S_{\geq}(g_1)\ket{\psi}$ will create some local excitation at $x_\mathrm{cut}$ and also at $x_\mathrm{cut}+M$ (which we will ignore).
% This local excitation can be annihilated by some unitary operator $U(g_1)^\dagger$ with support near $x_\mathrm{cut}$, 
% The twist phase is nothing but the charge of this excitation under a symmetry $g_2$, $T(g_1,g_2) = \langle  S(g_1)^\dagger U(g_2) S(g_1) U(g_1) \rangle_{H}$. (Verify this and make it clearer).

\section{Fractal Symmetries}\label{sec:fracsyms}

We can now discuss fractal symmetries.
The fractal symmetries we consider may be thought of as a combination of an on-site symmetry group imbued with some spatial structure.

Let us first consider a system with one fractal symmetry, described by the cellular automaton polynomial $f(x)$ over $\mathbb{F}_p$, which we will denote by 
\begin{equation}
    G = \mathbb{Z}_p^{(f,y)}
\end{equation}
which means that the on-site symmetry group is $\mathbb{Z}_p$, while the superscript, $(f,y)$,  denotes the associated spatial structure: $f$ denotes a cellular automaton described by the polynomial $f(x)$, and $y$ denotes the positive ``time'' direction of this cellular automaton (in this case, the positive $y$ direction).

Our systems have degrees of freedom placed on the sites of an $L_x\times L_y$ square lattice with periodic boundary conditions.
Each site is labeled by its index along the $x$ and $y$ direction, $(i,j)$, and transforms as an on-site linear representation $u_{ij}(g)$ under $g\in G$.  
For simplicity, we will only consider the cases where $L_x=p^N$ is a power of $p$, and $L_y$ chosen such that $f(x)^{L_y} = 1$.
The latter is not difficult to accomplish, as $f(x)^{L_x} = f(x^{L_x}) = f(1)$, so we may simply choose $L_y = k L_x>0$ such that $f(1)^{k}=1$.  
Note that reversibility of $f(x)$ implies $f(1)\neq 0$.

The symmetries of the system are in one-to-one correspondence with valid space-time histories of the cellular automaton.  
The choices of $L_x$ and $L_y$ made earlier means that any state $\vec{a}^{(0)}$ (on a ring of circumference $L_x$) is cyclic in time with period dividing $L_y$: $\vec{a}^{(L_y)} = \vec{a}^{(0)}$.
Given a valid trajectory $\vec{a}^{(j)}$, the operator $\prod_{ij} u_{ij}(g^{a_{i}^{(j)}})$ for $g\in G$ represents a valid symmetry operator.
The entire space-time trajectory $\vec{a}^{(j)}$ is determined solely by its state at a particular time $j_0$, $\vec{a}^{(j_0)}$, which can be in any of $p^{L_x}$ states.
The total symmetry group will therefore be given by $G_\mathrm{tot}=(\mathbb{Z}_p)^{L_x}$.

Let us identify a particular element $\mathit{g}$ as a generator for $\mathbb{Z}_p$.
Then, let a set of $L_x$ generators for $G_\mathrm{tot}=\mathbb{Z}_p^{L_x}$, defined \emph{with respect to $j_0$}, be $\{\mathit{g}_i^{(j_0)}\}_{0\leq i < L_x}$.
We may then define a vectorial representation of group elements via the one-to-one mapping from vectors $\vec{v}\in\mathbb{F}_p^{L_x}$ to group elements, 
\begin{equation}
    \mathit{g}^{(j_0)}[\vec{v}] = \prod_{i=0}^{L_x-1} (\mathit{g}^{(j_0)}_i)^{v_i} \in G_\mathrm{tot}
\end{equation}
The action of each of these symmetry elements on the system is defined as
\begin{equation}
    S(\mathit{g}^{(j_0)}[\vec{v}]) = \prod_{j=0}^{L_x-1} u_j[\mathit{g};\matr{F}^{j-j_0} \vec{v}]
\end{equation}
where we have introduced the vectorial representation for $u_{ij}(g)$ on a row $j$,
\begin{equation}
    u_j[\mathit{g};\vec{v}] \equiv \prod_{i=0}^{L_x-1} u_{ij}(\mathit{g}^{v_i})
\end{equation}
Thus, $S(\mathit{g}^{(j_0)}[\vec{v}])$ is the unique symmetry operator that acts as $u_j(\mathit{g})[\vec{v}]$ on the row $j_0$.
It can be viewed as the symmetry operation corresponding to the space-time trajectory of a CA which is in the state $\vec{v}$ at time $j_0$.
Because $f(x)^{L_y}=1$ due to our choice of $L_x$ and $L_y$, any initial state is guaranteed to come back to itself after time $L_y$, representing a valid cyclic space-time trajectory.

\begin{figure}
    \centering
    \includegraphics[width=0.40\textwidth]{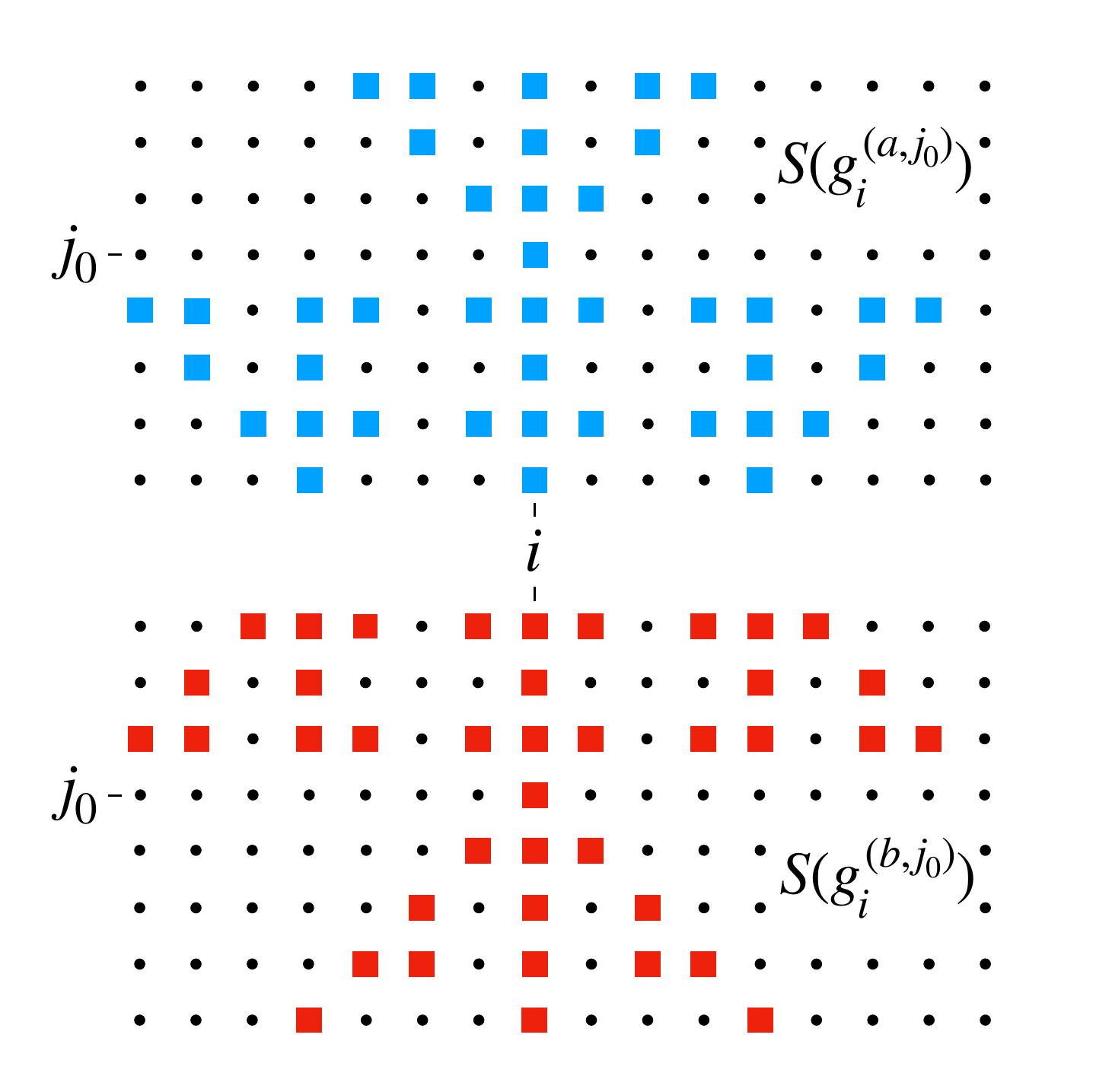}
    \caption{Example of a symmetry generator (top) $S(g_i^{(a,j_0)})$ or (bottom) $S(g_i^{(b,j_0)})$ for the fractal generated by $f(x)=\bar x+1+x$ with $p=2$.
    Sites with blue or red squares are acted on by $u_{ij}(\mathit{g}^{(a)})$ or  $u_{ij}(\mathit{g}^{(b)})$, respectively, and form a valid space-time trajectory of a cellular automaton.
    }
    \label{fig:basischoice}
\end{figure}

We may choose as a generating set the operators defined with respect to row $j_0$,
\begin{equation}
    S(\mathit{g}^{(j_0)}[\vec{e}_i])
    = S(\mathit{g}_i^{(j_0)})
\end{equation}
where $\vec{e}_i$ is the unit vector $(\vec{e}_i)_{i^\prime} = \delta_{i i^\prime}$.
These act on only a single site on the row $j_0$, and an example of which is shown in Figure~\ref{fig:basischoice} (top).
However, notice that this choice of basis is only ``most natural'' when viewed on the row $j_0$.
Suppose we wanted to change the row which we have defined our generators with respect to from $j_0$ to $j_1$.
How are the new operators related to our old ones?  
Well, one can readily show that
\begin{align}
    S(\mathit{g}^{(j_1)}[\vec{v}]) &= \prod_j u_j[\mathit{g};\matr{F}^{j-j_1}\vec{v}]\\
     &= \prod_j u_j[\mathit{g};\matr{F}^{j-j_0}\matr{F}^{j_0-j_1}\vec{v}] \\
     &= S(\mathit{g}^{(j_0)}[\matr{F}^{j_0-j_1}\vec{v}])
\end{align}
is simply related via multiplication of $\vec{v}$ by powers of $\matr{F}$.
Thus, 
\begin{equation}
    \mathit{g}^{(j_1)}[\vec{v}] = 
    \mathit{g}^{(j_0)}[\matr{F}^{j_0-j_1}\vec{v}] 
\label{eq:vectorevolve}
\end{equation}

In general, we can have systems with multiple sets of fractal symmetries.  
The other main situation we consider is the case of two fractal symmetries of the form
\begin{equation}
    G = \mathbb{Z}_p^{(f,y)} \times \mathbb{Z}_p^{(\bar{f},\bar{y})}
\end{equation}
where $\bar{x}\equiv x^{-1}$ and $\bar{f}(x) \equiv f(\bar{x})$.
This is the form of fractal symmetry known to protect non-trivial fractal SPTs~\cite{Devakul2018-ru,Kubica2018-dp}.
The first fractal represents a CA evolving in the positive $y$ direction with the rule $f(x)$, and the second represents a CA evolving in the opposite $y$ direction with the rule $\bar{f}(x)$ (they are spatial inversions of one another).
In this case, we have one generator from each $\mathbb{Z}_p$, $\mathit{g}^{(a)}$ and $\mathit{g}^{(b)}$, and we can define two sets of fractal symmetry generators as above with respect to a row $j_0$.
Let us call the two sets of generators $\{\mathit{g}^{(a,j_0)}_i\}_i$ and $\{\mathit{g}^{(b,j_0)}_i\}_i$, and define their corresponding vectorial representation.
A general $a$ or $b$ type symmetry acts as
\begin{align}
\begin{split}
    S(\mathit{g}^{(a,j_0)}[\vec{v}]) &= \prod_{j=0}^{L_x-1} u_j[\mathit{g}^{(a)};\matr{F}^{j-j_0} \vec{v}]\\
    S(\mathit{g}^{(b,j_0)}[\vec{v}]) &= \prod_{j=0}^{L_x-1} u_j[\mathit{g}^{(b)};(\matr{F}^T)^{j_0-j} \vec{v}]
    \end{split}
\end{align}
where we have used the fact that the matrix form of $\bar{f}(x)$ is given by $\matr{F}^T$.
A generator for an $a$ and a $b$ type symmetry are shown in Figure~\ref{fig:basischoice}.
The generalization of Eq~\ref{eq:vectorevolve} for moving to a new choice of basis $j_1$ for an $a$ or $b$ type symmetry is
\begin{align}
\begin{split}
    \mathit{g}^{(a,j_1)}[\vec{v}] &= 
    \mathit{g}^{(a,j_0)}[\matr{F}^{j_0-j_1}\vec{v}]\\
    \mathit{g}^{(b,j_1)}[\vec{v}] &= 
    \mathit{g}^{(b,j_0)}[(\matr{F}^{T})^{j_1-j_0}\vec{v}]
    \end{split}
    \label{eq:gabevolve}
\end{align}
\section{Local phases}\label{sec:localphases}

Consider performing the symmetry twisting experiment on a system with fractal symmetries.
We can view the system as a cylinder with circumference $L_x$ and consider twisting the symmetry as discussed in Sec~\ref{sec:1dtwist}.
We separately discuss the cases of one or two fractal symmetries of a specific form first, and then go on to more general combinations.
Our main findings in this section are summarized as:
\begin{enumerate}
\item For the case of one fractal symmetry, $G=\mathbb{Z}_p^{(f,y)}$, no non-trivial SPT phases may exist
\item For the case of two fractal symmetries, $G=\mathbb{Z}_p^{(f,y)}\times\mathbb{Z}_p^{(\bar{f},\bar{y})}$, if we only allow for locality up to some lengthscale $\ell$, then there are a only a finite number of possible SPT phases (scaling exponentially in $\ell^2$)
\item For the case of more fractal symmetries, it is sufficient to identify pairs of symmetries of the form $\mathbb{Z}_p^{(f,y)}\times\mathbb{Z}_p^{(\bar{f},\bar{y})}$, and apply the same results from above.
\end{enumerate}

\subsection{One fractal symmetry}
\label{subsec:localone}

Let us take $G=\mathbb{Z}_p^{(f,y)}$ and consider twisting by a particular element $g^{(j_0)}_i \in G_\mathrm{tot}$.  
Since the twist phase doesn't depend on the position of the cut, we can choose to make the cut on the row $j_\mathrm{cut}=j_0$.  
The twisted Hamiltonian $H_\mathrm{twist}(g_i^{(j_0)})$ is then obtained by conjugating terms in the Hamiltonian which cross $j_\mathrm{cut}$ by the truncated symmetry operator $S_{\geq}(g_i^{(j_0)})$.

Let the Hamiltonian be written as a sum 
\begin{equation}
    H = \sum_{i,j} H_{ij}
\end{equation}
where each $H_{ij}$ is a local term with support near site $(i,j)$.
Now, consider twisting the Hamiltonian by $g_i^{(j_0)}$ across the cut which also goes along the row $j_0$.  
As can be seen in Figure~\ref{fig:fractwist} (left), $S_{\geq}(g_i^{(j_0)})$ acts on a single site on row $j_0$, and extends into the fractal structure on the rows above.  
The important point is that $S_{\geq}(g_i^{(j_0)})$ \emph{only acts differently from an actual symmetry operator at the point $(i,j_0)$} (and on some row $j_0+R$ far away).
Thus, the twisted Hamiltonian may be written as
\begin{equation}
    H_\mathrm{twist}(g_i^{(j_0)})\ket{\psi} = U_{g_i^{(j_0)}} H U_{g_i^{(j_0)}}^\dagger\ket{\psi}
\end{equation}
when acting on the ground state $\ket{\psi}$, 
for some unitary $U_{g_i^{(j_0)}}$ with support near the site $(i,j_0)$.
Note that there is always some freedom in choosing this unitary.

Then, consider measuring the charge of a symmetry  $g_{i^\prime}^{(j_0-l_y)}$ in response to this twist, as in Eq~\ref{eq:meascharge}.
Clearly, only those symmetry operators whose support overlaps with the support of $U_{g_i^{(j_0)}}$ may have picked up a charge.
Suppose the support of every $U_{g_i^{(j_0)}}$ is bounded within some $(2l_x+1)\times (2l_y+1)$ box centered about $(i,j_0)$, such that only sites $(i^\prime,j^\prime)$ with $|i^\prime-i|\leq l_x$ and $|j^\prime-j_0|\leq l_y$ lie in the support.
As can be seen in Figure~\ref{fig:fractwist} (left), $S(g_{i^\prime}^{(j_0-l_y)})$ only overlaps with this box for $i^\prime$ in the range
\begin{equation}
    -l_x- 2l_y k_b \leq i^\prime-i \leq l_x- 2l_y k_a
    \label{eq:diagband}
\end{equation}
and therefore, $\Omega(g_{i^\prime}^{(j_0-l_y)}, g_{i}^{(j_0)})$ may only be non-trivial if $i^\prime-i$ is within some small range.
This places a constraint on the allowed twist phases.
In addition, this must be true \emph{for all choices of $j_0$}.  
It turns out this is a \emph{very} strong constraint, and eliminates all but the trivial phase in the case of $G=\mathbb{Z}_p^{(f,y)}$, and only allows a finite number of specific solutions for the case $G=\mathbb{Z}_p^{(f,y)}\times \mathbb{Z}_p^{(\bar f, \bar y)}$, as we will show.

We also do not strictly require that the support of $U_{g_i^{(j_0)}}$ be bounded in a box.  
This will generally not be the case, as the operator may have an exponentially decaying tail.
Consider a unitary $U$ which has some nontrivial charge $e^{i\phi}\neq 1$ under $S$, meaning
\begin{equation}
     S U S^\dagger = e^{i\phi} U
\end{equation}
when acting on the ground state.
Clearly, if the support of $U$ and $S$ are disjoint, this cannot be true.
Next, consider any decomposition of $U$ into a sum of matrices $U_k$, 
$ U = \sum_k U_k $,
and suppose that some of the $U_k$ had disjoint support with $S$.  
Then, we may write
\begin{equation}
    U = \sum_{k \in \overline{\mathcal{D}}} U_k + \sum_{k\in \mathcal{D}} U_k
\end{equation}
where $k\in\mathcal{D}$ are all the $k$ for which $U_k$ and $S$ have disjoint support, and $k\in\overline{\mathcal{D}}$ are all the $k$ for which they do not.
But then
\begin{align}
    S U S^\dagger &= \sum_{k \in \overline{\mathcal{D}}} S U_k S^\dagger + \sum_{k\in\mathcal{D}} U_k\\
    &\neq e^{i\phi} U 
\end{align}
as the disjoint component has not picked up a phase $e^{i\phi}$, and
$SU_k S^\dagger$ for $k\in\mathcal{D}$ cannot have disjoint support with $S$ (since only identity maps to identity under unitary transformations) and so can't affect the disjoint component of $U$.
Thus, let us define a subset of sites, $\mathcal{A}(U)$, defined as
\begin{equation}
    \mathcal{A}(U) = \bigcap_{\substack{\mathrm{decomps}\\U=\sum_k U_k} }
    \bigcap_k \mathrm{Supp}(U_k)
    \label{eq:mcAdef}
\end{equation}
where the first intersection is over all possible decompositions $U=\sum_k U_k$, and $\mathrm{Supp}(U_k)$ is the support of $U_k$ (the subset of sites for which it acts as non-identity).
$U$ can only have nontrivial charge under $S$ if $\mathcal{A}(U)$ overlaps with the support of $S$.
In our case, $l_x$ and $l_y$ should actually be chosen such that $\mathcal{A}(U_{g_i^{(j_0)}})$ may always be contained within the $(2l_x+1,2l_y+1)$ box.
An exponentially decaying tail of $U$ is therefore completely irrelevant, as $\mathcal{A}(U)$ only cares about the smallest part, before the decay begins.
The exact value of $l_x$ or $l_y$ is not too important --- what is important is that it is finite and small.

We also note that the twist phases obtained when twisting along a cut in the $y$ direction will be different, but are not independent of our twist phases for a cut along the $x$ direction.
To see why this is, consider a truncated symmetry operator which has been truncated by a cut in the $y$ direction.  
This may alternatively be viewed as an untruncated symmetry operator, multiplied by $S_{\geq}(g_i^{(j)}$ at various $(i,j)$s located near the cut.  
The action of twisting this symmetry for a cut along the $y$ direction is then also fully determined by the same set of $U_{g_i^{(j)}}$ from before, and is therefore not independent of the twist phases for a cut along the $x$ direction.
Thus, it is sufficient to examine only the set of twist phases for a cut parallel to $x$, as we have been discussing.
As we chose $y$ to be the ``time'' direction of our CA, twisting along the $x$ direction is far more natural.

Let us make some definitions which will simplify this discussion.
Notice that $\Omega(g^{(j_0)}[\vec{v}],g^{(k_0)}[\vec{w}])$ may be described by the bilinear form $\mathbb{F}_p^{L_x}\times \mathbb{F}_p^{ L_x} \rightarrow \mathbb{F}_p$ represented by the skew-symmetric matrix $\matr{W}^{(j_0,k_0)}\in\mathbb{F}_p^{L_x\times L_x}$ defined according to
\begin{equation}
    \Omega(g^{(j_0)}[\vec{v}], g^{(k_0)}[\vec{w}]) = e^{\frac{2\pi i}{p}  \vec{v}^T \matr{W}^{(j_0,k_0)} \vec{w}}
\end{equation}
and that $\matr{W}^{(j_0, k_0)}$ for any $(j_0,k_0)$ contains full information of the twist phases.  
Furthermore, since $g^{(j_1)}[\vec{v}] = g^{(j_0)}[\matr{F}^{j_0-j_1}\vec{v}]$, 
we can deduce that $\matr{W}$ transforms under this change of basis as
\begin{align}
        \matr{W}^{(j_1,k_1)} &= (\matr{F}^{j_0-j_1})^T \matr{W}^{(j_0,k_0)} \matr{F}^{k_0-k_1}
        \label{eq:Wevolvemat}
\end{align}

We say that a matrix $\matr{W}^{(j_0-l_y,j_0)}$, for a particular choice of $j_0$,  is \emph{local} if its only non-zero elements $W^{(j_0-l_y,j_0)}_{i^\prime i}\neq 0$ are within a small diagonal band given by Eq~\ref{eq:diagband}.
A stronger statement, which we will call \emph{consistent locality}, is that this is true for all $j_0$.
The matrix $\matr{W}^{(j_0-l_y,j_0)}$ for a physical state must be consistently local.

Let us adopt a polynomial notation which will be useful to perform computations.
We may represent the matrix $\matr{W}^{(j_0,k_0)}$ by a polynomial $W^{(j_0,k_0)}(u,v)$ over $\mathbb{F}_p$ as
\begin{equation}
    W^{(j_0,k_0)}(u,v) = \sum_{i i^\prime} W^{(j_0,k_0)}_{i^\prime i} u^{i^\prime} v^{i^\prime-i}
\end{equation}
with periodic boundary conditions $u^{L_x}=v^{L_x}=1$.
Locality is simply the statement that the powers of $v$ in this polynomial must be bounded by Eq~\ref{eq:diagband} (modulo $L_x$).
Now, consider what happens to this polynomial as we transform our basis choice from $j_0\rightarrow j_0-n$,
\begin{align}
    W^{(j_0-n-l_y,j_0-n)}(u,v) =  f(v)^n f(\bar{u}\bar{v})^n W^{(j_0-l_y,j_0)}(u,v)
\end{align}
which must be local for all $n$ if $W^{(j_0-l_y,j_0)}(u,v)$ is to be consistently local.

Let us start with $j_0=0$, and suppose that we have some $W^{(-l_y,0)}(u,v)$ that is non-zero and local. 
By locality, $W^{(-l_y,0)}(u,v)$ may always be brought to a form where the powers of $v$ are all within the range given by Eq~\ref{eq:diagband}.
Let $v^{a}$ and $v^{b}$ be the smallest and largest powers of $v$ in $W^{(l_y,0)}(u,v)$ once brought to this form, which must satisfy
\begin{equation}
    -l_x-2l_y k_b\leq a\leq b\leq l_x-2l_y k_a
\end{equation}
Now, consider $W^{(-l_y-n,-n)}(u,v)$ for small $n$,
\begin{align}
    W^{(-l_y-n,-n)}(u,v) =  f(v)^n f(\bar{u}\bar{v})^n W^{(l_y,0)}(u,v)
\end{align}
which (by adding degrees) will have $v^{a-n\delta_f}$ and $v^{b+n\delta_f}$ as the smallest and largest powers of $v$, where $\delta_f=\deg(x^{-k_a}f(x))>0$.
The smallest and largest powers will therefore keep getting smaller and larger, respectively, as we increase $n$.
Thus, there will always be some finite $n$ beyond which locality is violated, and so $W^{(-l_y,0)}(u,v)$ can never be consistently local.
The only consistently local solution is therefore given by $W^{(-l_y,0)}(u,v)=0$, which corresponds to the trivial phase.
We have therefore shown that no non-trivial local SPT phase can exist protected by only $G=\mathbb{Z}_p^{(f,y)}$ symmetry.

\begin{figure*}
    \centering
    \includegraphics[width=0.9\textwidth]{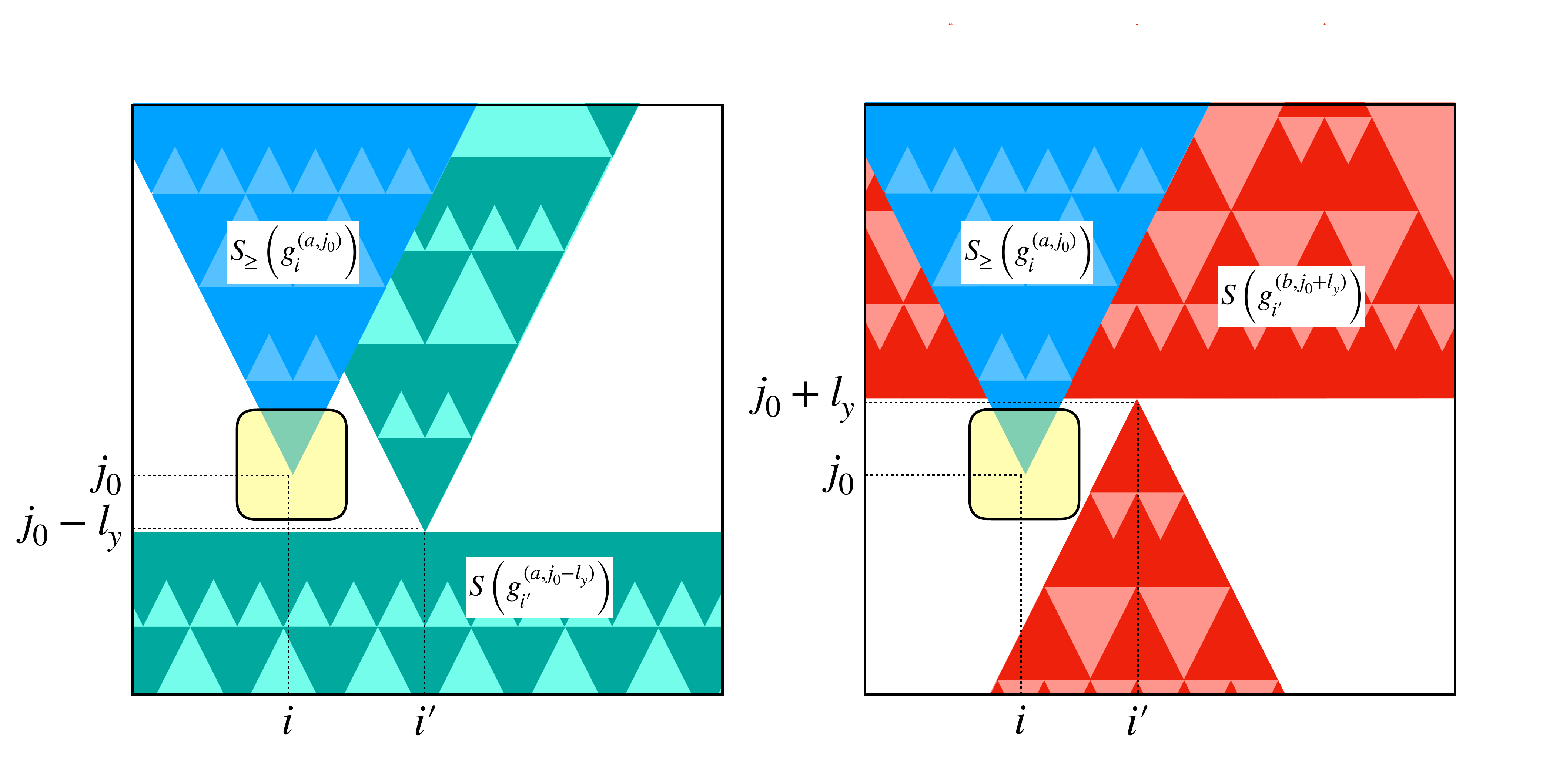}
    \caption{Measurement of the twist phases for (left) $\Omega(g^{(a,j_0-l_y)}_{i^\prime}, g^{(a,j_0)}_i)$ and (right) $\Omega(g^{(b,j_0+l_y)}_{i^\prime},g^{(a,j_0)}_i)$.
    Due to locality, the twist phase may only be non-trivial if the support of (left) $S(g^{(a,j_0-l_y)}_{i^\prime}$ or (right) $S(g^{(b,j_0+l_y)}_{i^\prime})$ has some overlap with the yellow box of size $(2l_x+1)\times(2l_y+1)$ about $(i,j_0)$.
    This implies that the twist phase must be trivial for  $i^\prime$ outside of a small region around $i$, a property which we call locality.
    However, this must be true for all choices of $j_0$, which greatly constrains the allowed twist phases.
    In the case of twist phases between the same type of symmetry (left), only the trivial set of twist phases, all $\Omega(g^{(a,j_0-l_y)}_{i^\prime}, g^{(a,j_0)}_i)=1$ is allowed.  
    Between an $a$ and a $b$ type symmetry (right), we show that only a finite number of solutions exist.
    }
    \label{fig:fractwist}
\end{figure*}

\subsection{Two fractal symmetries}\label{subsec:localtwo}
Let us now consider the more interesting case, $G=\mathbb{Z}_p^{(f,y)}\times \mathbb{Z}_p^{(\bar f,\bar y)}$, for which we know non-trivial SPT phases can exist.  
In this case, we have the symmetry generators $g_i^{(\alpha,j_0)}$ for $\alpha\in\{a,b\}$, and $0\leq i < L_x$.
As we showed in the previous section, the twist phase between two $a$ or two $b$ symmetries must be trivial.
The new ingredient comes in the form of non-trivial twist phases between $a$ and $b$ symmetries.

As can be seen in Fig~\ref{fig:fractwist}, by the same arguments as before, the twist phase
\begin{equation}
    \Omega(g_{i^\prime}^{(b,j_0+l_y)},g_{i}^{(a,j_0)})
\end{equation}
may only be non-trivial if
$i^\prime-i$ lies within some finite range,
\begin{equation}
    -l_x+2 l_y k_a \leq i^\prime-i\leq l_x + 2 l_y k_b.
    \label{eq:diagband2}
\end{equation}

Let us again define the matrix $\matr{W}^{(k_0,j_0)}$, but this time only between the $a$ and $b$ symmetries via
\begin{equation}
    T(g^{(b,k_0)}[\vec{w}], g^{(a,j_0)}[\vec{v}]) = e^{\frac{2\pi i}{p} \vec{w}^T\matr{W}^{(k_0,j_0)}\vec{v}}
\end{equation}
note that $\matr{W}^{(k_0,j_0)}$ need not be skew-symmetric like before.
From Eq~\ref{eq:gabevolve}, the changing of basis is given by
\begin{equation}
    \matr{W}^{(k_1,j_1)} = \matr{F}^{k_1-k_0} \matr{W}^{(j_0,j_1)} \matr{F}^{j_0-j_1}
\end{equation}
We are looking for matrices $W^{(j_0+l_y,j_0)}_{i^\prime i}$ which are \emph{local} (only non-zero within the diagonal band Eq~\ref{eq:diagband2}), and also \emph{consistently local}, meaning that this is true for all $j_0$.
Starting with $j_0=0$, then, we are searching for a local matrix $\matr{W}^{(l_y,0)}$, for which
\begin{equation}
    \matr{W}^{(l_y+n,n)} = \matr{F}^n \matr{W}^{(l_y,0)} \matr{F}^{-n}
    \label{eq:Wevolvemat2}
\end{equation}
is also itself local for all $n$.

Let us again go to a polynomial representation
\begin{equation}
    W^{(l_y,0)}(u,v) = \sum_{i^\prime i} W^{(l_y,0)}_{i^\prime i} u^{i^\prime} v^{i^\prime-i}
    \label{eq:Wpolydef}
\end{equation}
which leads to the relation
\begin{equation}
    W^{(l_y+n,n)}(u,v) =  f(v)^{-n} f(uv)^n W^{(l_y,0)}(u,v) 
\end{equation}
which must have only small (in magnitude) powers of $v$ for all $n$.  
However, $f(v)^{-1} \equiv f(v)^{L_y-1}$ contains arbitrarily high powers of $v$, and therefore simply adding degrees as before does not work and we may expect that a generic $W^{(l_y,0)}(u,v)$ will become highly non-local immediately.
Instead, what must be happening is that, at each step, $f(uv)W^{(j_0+l_y,j_0)}(u,v)$ must contain some factor of $f(v)$ (when viewed as a polynomial without periodic boundary conditions) such that the $f(v)^{-1}$ can divide out this factor cleanly, producing a local $W^{(j_0+1+l_y,j_0+1)}(u,v)$.

How does this work in the case of the known fractal SPT~\cite{Devakul2018-ru}?  In that case, $\matr{W}^{(0,0)}$ is already local and is given by the identity matrix.  
Then, clearly $\matr{W}^{(n,n)} = \matr{W}^{(0,0)}$ as it is invariant under Eq~\ref{eq:Wevolvemat2}, and remains local for all $n$.
In the polynomial language, the identity matrix corresponds to the polynomial $W^{(0,0)}(u,v) = \sum_{i} u^i$, which has the property of translation invariance: $W^{(0,0)}(u,v) = u W^{(0,0)}(u,v)$.
In this case, $f(uv)W^{(0,0)}(u,v) = f(v)W^{(0,0)}(u,v)$, and so can be safely multiplied by $f(v)^{-1}$.
In fact, any translation invariant solution, $W^{(0,0)}(u,v) = g(v)\sum_{i} u^i $ for any $g(v)$, is invariant under multiplying by $f(v)^{-1}f(uv)$.
% Thus, this is a case of how locality can enforce a kind of translation invariance of the twist phases.

\begin{figure*}
    \centering
    \includegraphics[width=\textwidth]{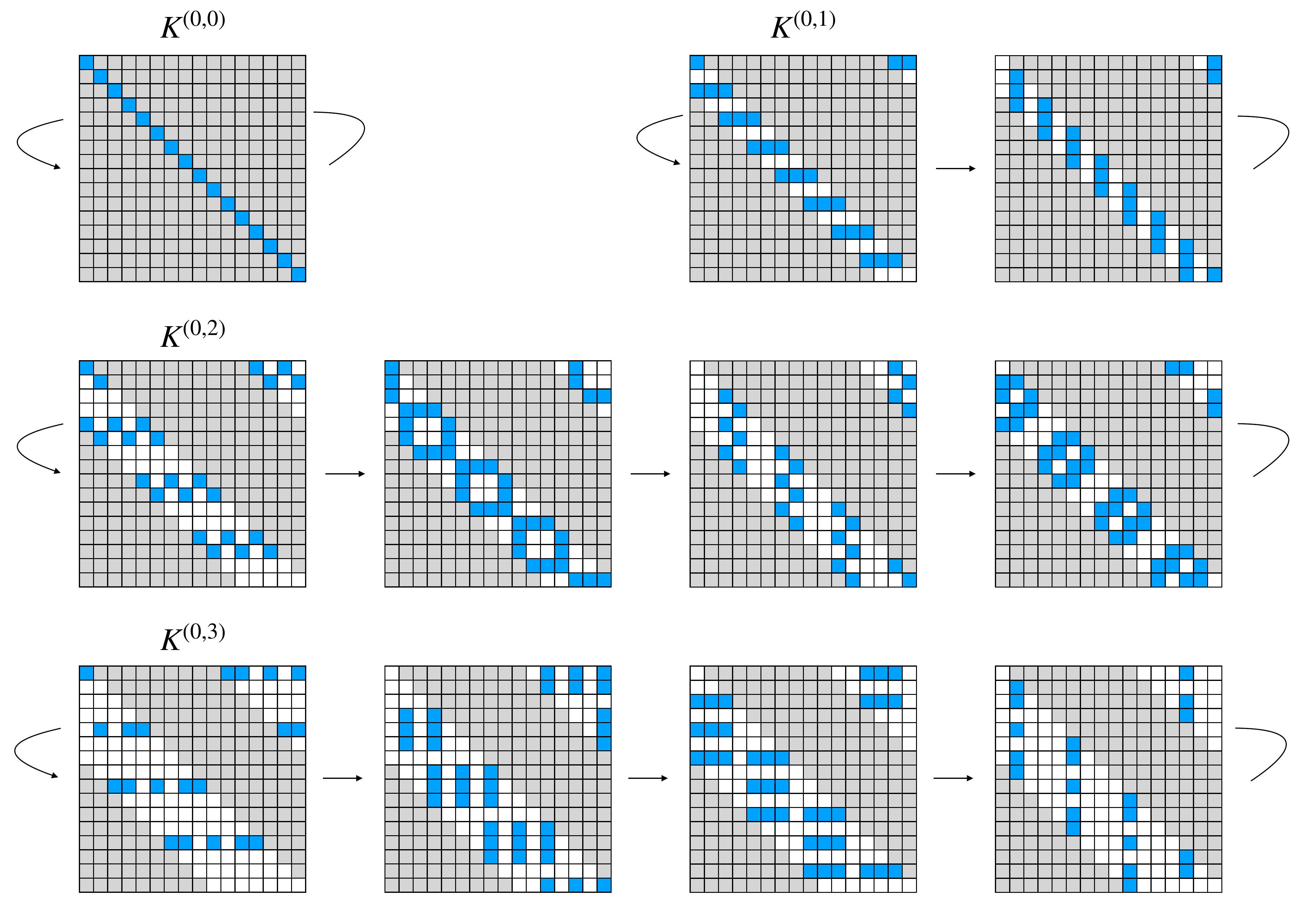}%
    \caption{Visualization of the matrices $\matr{K}^{(k,m)}$ for the example of $f(x)=1+x+x^2$ and $p=2$, for $k=0$ and  $m=0,1,2,3$ (other $k$ can be obtain by shifting every element $k$ steps to the left).
    Each blue cell $(i^\prime,i)$, counting from the top left, represents a non-zero matrix element $K^{(0,m)}_{i^\prime i}=1$.  
    The arrows indicate evolution by
    $\matr{K}\rightarrow\matr{F} \matr{K}\matr{F}^{-1}$, under which they exhibit cycles of period $2^{\lceil \log_2 (m+1)\rceil}$, as can be seen.
    Each of them are only non-zero about a small diagonal band (non-gray squares) of width given by $D_m = 1+2m$.
    A $\matr{K}^{(k,m)}$ is local if this white band fits inside some diagonal band (Eq~\ref{eq:diagband2}).
    If a $\matr{K}^{(k,m)}$ is local, then it can be seen that under evolution it retains locality (the white band never gets larger), a property which we call consistent locality.
    The main result of this paper is that any consistently local matrix can be written as a linear sum of local $\matr{K}^{(k,m)}$.
    Since there are only a finite number of local $\matr{K}^{(k,m)}$, there are only a finite number of consistently local matrices that can be written, and therefore a finite number of distinct phases in the thermodynamic limit.
    The number will depend on the choice of $(l_x,l_y)$, i.e. how local the model is.
    Notice that consistent locality is non-generic: if we just pick a local matrix by filling in elements along the diagonal band at random, it will generically quickly become highly non-local after a few steps of evolution.
    }
    \label{fig:cycles}
\end{figure*}

We now state the main result of this paper: 
a special choice of basis functions $v^k\mathcal{K}_m(u,v)$ with the property that
 $W(u,v)$ is consistently local if and only if in the unique decomposition
\begin{equation}
    W(u,v) = \sum_{k=0}^{L_x-1} \sum_{m=0}^{L_x-1} C_{k,m} v^k \mathcal{K}_m(u,v)
\end{equation}
where $C_{k,m}\in \mathbb{F}_p$ are constants, each $C_{k,m}v^k \mathcal{K}_m(u,v)$ is itself individually local.  
$\mathcal{K}_m(u,v)$ is given by
\begin{align}
    \mathcal{K}_m(u,v) &= (u-1)^{L_x-1-m} 
    \mathcal{V}_m(v) \\
    \mathcal{V}_m(v) &= \prod_{i=0}^{N_f-1} f_i(v)^{\overline{m}_i}
    \label{eq:mcKdef}
\end{align}
where $f_i(x)$ are the $N_f$ unique irreducible factors of the polynomial $\tilde{f}(x) \equiv  x^{-k_a}f(x)$ appearing $r_i$ times, $\tilde{f}(x) = \prod_{i} f_i(x)^{r_i}$,
and 
$\overline{m}_i = \lfloor m/p^{\alpha_i}\rfloor p^{\alpha_i}$ where $\alpha_i$ is the power of $p$ in the prime decomposition of $r_i$.
The proof of this is rather technical and is delegated to Section~\ref{sec:proof}.
Thus, any phase can simply be constructed by finding all $v^k\mathcal{K}_m(u,v)$ that are local, and choosing their coefficients $C_{k,m}$ freely.

Let us go back to the matrix representation, and define the corresponding matrices $\matr{K}^{(k,m)} \leftrightarrow v^k \mathcal{K}_m(u,v)$, following the same mapping as Eq~\ref{eq:Wpolydef}.
The elements of the matrix $K^{(k,m)}_{i^\prime i}$ are non-zero if and only if $k\leq i^\prime-i \leq k+D_m$, where $D_m$ is the degree of $\mathcal{V}_m(v)$.
$D_m$ increases monotonically with $m$, and is bounded by $D_m\leq m \delta_f $.
This bound is saturated when $v^{-k_a}f(x)$ is a product of irreducible polynomials, each of which appear only once.
Our main result (Theorem~\ref{theorem}) states that \emph{any} consistently local $\matr{W}^{(l_y,0)}$ can be written as a linear sum of local $\matr{K}^{(k,m)}$.
Thus, it is straightforward to enumerate all possible $\matr{W}^{(l_y,0)}$, which is simply all matrices in the subspace of $\mathbb{F}_p^{L_x\times L_x}$ spanned by the set of local $\matr{K}^{(k,m)}$ (note that the full set of $\{\matr{K}^{(k,m)}\}_{km}$ for all $0\leq k,m<L_x$ forms a complete basis for this space).
Figure~\ref{fig:cycles} shows $\matr{K}^{(0,m)}$ for $m=0,1,2,3$ for a specific example, and how they evolve from one row to the next while maintaining locality.

A property of the matrices $\matr{K}^{(k,m)}$ is that they are periodic with period $p^{N_m}$, meaning $K^{(k,m)}_{i+p^{N_m},i^\prime+p^{N_m}}=K^{(k,m)}_{i i^\prime}$, where $N_m \equiv \lceil \log_p (m+1)\rceil$.
They also have cycles of period $p^{N_m}$, meaning $\matr{K}^{(k,m)} = \matr{F}^{p^{N_m}} \matr{K}^{(k,m)} \matr{F}^{-p^{N_m}}$.
Since $D_m$ increases monotonically with $m$,
 only $m$ up to some maximum value, $M$, are local and may be included in $\matr{W}^{(l_y,0)}$.
 We therefore see that $\matr{W}^{(l_y,0)}$ must be periodic with period $p^{N_M}$.  
Thus, locality enforces that the projective representation characterizing the phase, $\matr{W}^{(l_y,0)}$, be $p^{N_M}$-translation invariant!
This is a novel phenomenon that does not appear in, say, subsystem SPTs with line-like symmetries where the projective representation does not have to be periodic (and as a result there are an infinite number of distinct phases in the thermodynamic limit, even with a local model).

How many possible phases may exist for a given $(l_x, l_y)$?
This is given by the number of $\matr{K}^{(k,m)}$ that can fit within a diagonal band of width $\ell\equiv1+2l_x + 4l_y\delta_f$.
For each $m$, $\matr{K}^{(k,m)}$ is local if $0\leq k < \ell-D_m $.
Thus, there are $C_m\equiv \max\{\ell-m D_m,0\}$ possible $k$ values for each $m$.
The total number of local $\matr{K}^{(k,m)}$ is then $\sum_m C_m$.

Consider the case where $f(x) = x^{k_a} f_1(x) f_2(x) \dots$ where each unique irreducible factor $f_i(x)$ only appears once.  
Then, $D_m = m\delta_f$.
The total number of local $\matr{K}^{(k,m)}$ is then
\begin{align}
N_{loc} &= \sum_{m=0}^{\infty} \max\left\{\ell-m \delta_f,0\right\} \\
&= \frac{\delta_f}{2} \lceil C \rceil \left( 2 C - \lceil C \rceil + 1 \right)
\end{align}
where we have just summed $m$ to infinity since $\ell\ll L_x/\delta_f$,  and $C \equiv \ell/\delta_f$. 
Notice that $N_{loc}$ only depends only on the combination $\ell$, and not specifically what $l_x$ and $l_y$ are.
% If we take $l_x=l_y=l$, this scales as $N\sim\mathcal{O}(l^2)$.
The $\matr{W}^{(l_y,0)}$ describing each phase is therefore a linear sum of these $N_{loc}$ matrices $\matr{K}^{(k,m)}$, and so
the total number of possible phases is $p^{N_{loc}}$.
These phases are in one-to-one correspondence with elements of the group $\mathbb{Z}_p^{N_{loc}}$, and
exhibit the group structure under stacking.
Note that this number is an \emph{upper bound} on the number of possible phases with a given $(l_x,l_y)$.

Consider the example in Figure~\ref{fig:cycles}, which has $f(x)=1+x+x^2$ and $p=2$.
Suppose we were interested in phases that have locality $(l_x,l_y)=(1,0)$, then the matrix $W^{(l_y,0)}_{i^\prime i}$ may only be non-zero if $-1\leq i^\prime-i\leq 1$.  
The only local $\matr{K}^{(k,m)}$ matrices are then $\matr{K}^{(-1,0)}$, $\matr{K}^{(0,0)}$, $\matr{K}^{(1,0)}$, and $\matr{K}^{(-1,1)}$.  
Then, our result (Theorem~\ref{theorem}) states that \emph{all} consistently local $\matr{W}^{(l_y,0)}$ are a linear sum (with binary coefficients) of these four $\matr{K}^{(k,m)}$ matrices.  
There are therefore only $2^{4}$ possible phases, and they all have twist phases that are periodic with a period of $2$ sites (or $1$ if the coefficient of $\matr{K}^{(-1,1)}$ is zero).

\subsection{More fractal symmetries}
Beyond these two cases, we may imagine more general combinations of fractal symmetries of the form
\begin{equation}
    G = \prod_{i=0}^{N-1} \mathbb{Z}_{p}^{(f_i, y^{\eta_i})}
    \label{eq:generalG}
\end{equation}
where we have $N$ different fractals $f_i(x)$, which each have positive time direction $y^{\eta_i}$ given by $\eta_i=\pm 1$.
We again assume none of $f_i(x)$ are monomials.
In this language, the previous case of two fractal symmetries is given by $N=2$ with $f_0(x)=f_1(\bar{x}) = f(x)$, and $\eta_0=-\eta_1=1$.
Note that we could have allowed $p$ to vary among the fractals --- the reason we do not consider this case is that the twist phases between generators of $\mathbb{Z}_p$ and $\mathbb{Z}_{q}$ are $\gcd(p,q)$th roots of unity, but since $p$ and $q$ are both prime, this phase must be trivial.

By an argument similar to that given in Sec~\ref{subsec:localone}, we may show that any twist phase between the two generators of $\mathbb{Z}_p^{(f_i,y^{\eta_i})}$ and $\mathbb{Z}_p^{(f_j,y^{\eta_j})}$ for which $\eta_i=\eta_j$ must be trivial.  
What about when $\eta_i\neq \eta_j$?  
In that case, we can show that there may only exist non-trivial twist phases between them if $f_i(x) = f_j(\bar{x})$.

Suppose we have some matrix $\matr{W}^{(l_y,0)}$ describing twist phases between symmetry generators 
of $\mathbb{Z}_p^{(f_i,y^{\eta_i})}$ and $\mathbb{Z}_p^{(f_j,y^{\eta_j})}$.
Going to a polynomial representation $W^{(l_y,0)}(u,v)$ (as in Eq~\ref{eq:Wpolydef}), the change of basis to a different row is
\begin{equation}
    W^{(l_y+n,n)}(u,v) = f_i(uv)^{n} f_j(\bar{v})^{-n} W^{(l_y,0)}(u,v)
\end{equation}
which must be local for all $n$.
Suppose that $W^{(l_y,0)}(u,v)$ is $p^{k}$-periodic, such that $u^{p^k}W^{(l_y,0)}(u,v) = W^{(l_y,0)}(u,v)$.
Then, 
\begin{equation}
    W^{(l_y+np^k,np^k)}(u,v) = (f_i(v)^{p^k} f_j(\bar{v})^{-p^k})^n W^{(l_y,0)}(u,v)
\end{equation}
should also be local for all $n$ (recall that locality in the polynomial language is a statement about the powers of $v$ present).
This implies that $f_i(v)^{p^k} f_j(\bar{v})^{-p^k} = 1$, or $f_i(v^{p^k}) =  f_j(\bar{v}^{p^k})$.
If $p^k \ll L_x$, then this means that we must have $f_i(x) = f_j(\bar{x})$.
In the case where $p^k\not\ll L_x$, we may simply consider larger system sizes $L_x,L_y\rightarrow p^m L_x, p^m L_y$, but with the same periodicity $p^k$, and come to the same conclusion.
Hence, there can only exist non-trivial twist phases between symmetries with $f_i(x) = f_j(\bar{x})$ and $\eta_i = -\eta_j$.

For the more general group $G$ in Eq~\ref{eq:generalG}, to find all the possible phases with a fixed locality $(l_x,l_y)$, we should simply find all pairs $(i,j)$ where $\eta_i = -\eta_j$ and $f_i(x)=f_j(\bar{x})$, and construct a local $\matr{W}^{(l_y,0)}$ matrix for each such $(i,j)$ pair.
Thus, the case with two fractal symmetries $G=\mathbb{Z}_p^{(f,y)}\times\mathbb{Z}_p^{(\bar{f},\bar{y})}$ already contains all the essential physics.

\section{Constructing commuting models for arbitrary phases}\label{sec:construct}
In this section, we show that it is indeed possible to realize all the phases derived in the previous section for systems with two fractal symmetries, $G=\mathbb{Z}_p^{(f,y)}\times\mathbb{Z}_p^{(\bar f,\bar y)}$, in local Hamiltonians.
We show how to construct a Hamiltonian, composed of mutually commuting local terms, for an arbitrary phase characterized by the matrix $\matr{W}^{(l_y,0)}$.
These Hamiltonians are certainly not the \emph{most} local models that realize each phase, but they are quite conceptually simple and the construction works for any given $\matr{W}^{(l_y,0)}$.

Let us define $\mathbb{Z}_p$ generalizations of the Pauli matrices $X$ and $Z$ satisfying the following algebra,
\begin{align}
    X^N = Z^N = 1\\
    XZ = e^{\frac{2\pi i}{p}}ZX
\end{align}
and may be represented by a $p\times p$ diagonal matrix $Z$ whose diagonals are $p$-th roots of unity, and $X$ as a cyclic raising operator in this basis.

The local Hilbert space on each site $(i,j)$ of the square lattice is taken to be two such $p$-state degrees of freedom, labeled $a$ and $b$, which are operated on by the operators $Z_{ij}^{(\alpha)}$ and $X_{ij}^{(\alpha)}$, for $\alpha\in\{a,b\}$.
Each $Z_{ij}^{(\alpha)}$ only has non-trivial commutations with $X_{ij}^{(\alpha)}$ on the same site and $\alpha$.

Let us also define a vectorial representation of such operators: functions of vectors $\vec{v}\in\mathbb{F}_p^{L_x}$ to operators on the row $j$ as
\begin{align}
    Z_j^{(\alpha)}[\matr{v}] = \prod_{i=0}^{L_x-1} (Z_{ij}^{(\alpha)})^{v_i}\\
    X_j^{(\alpha)}[\matr{v}] = \prod_{i=0}^{L_x-1} (X_{ij}^{(\alpha)})^{v_i}
\end{align}
One can verify that the commutation relations in this representation are
\begin{equation}
X_j^{(\alpha)}[\vec{v}] Z_{j}^{(\alpha)}[\vec{w}] 
=
e^{\frac{2\pi i}{p} \vec{v}^T\vec{w}}
Z_{j}^{(\alpha)}[\vec{w}] X_j^{(\alpha)}[\vec{v}] 
\end{equation}
for two operators on the same row $j$ with the same $\alpha\in\{a,b\}$, and trivial otherwise.

The onsite symmetry group is $G=\mathbb{Z}_p^{(f,y)}\times \mathbb{Z}_p^{(\bar f, \bar y)}$.
Let us label the first $\mathbb{Z}_p$ factor as $a$, and the second as $b$, and let $\mathit{g}^{(a)}$ and $\mathit{g}^{(b)}$ be generators for the two.
Then, we take the onsite representation 
\begin{equation}
    u_{ij} (\mathit{g}^{(\alpha)}) = X^{(\alpha)}_{ij} 
\end{equation}
As always, we take $L_x$ to be a power of $p$, and $L_y$ such that $f(x)^{L_y}=1$.  
The total symmetry group is $G_\mathrm{tot}=\mathbb{Z}_p^{L_x}\times\mathbb{Z}_p^{L_x}$.
An arbitrary element of the first $\mathbb{Z}_p^{L_x}$ factor, with basis defined with respect to row $j_0$, is given by
\begin{equation}
S(\mathit{g}^{(a,j_0)}[\vec{v}]) = \prod_{j=0}^{L_y-1} X^{(a)}_j[\matr{F}^{j-j_0} \vec{v}]
\end{equation}
and of the second by
\begin{equation}
S(\mathit{g}^{(b,j_0)}[\vec{v}]) = \prod_{j=0}^{L_y-1} X^{(b)}_j[(\matr{F}^T)^{j_0-j} \vec{v}]
\end{equation}

Suppose we are given a consistently local matrix $\matr{W}^{(l_y,0)}$ representing the twist phase.
For convenience, let us denote $\matr{W}_j\equiv \matr{W}^{(l_y+j,j)}$.
Recall that consistent locality implies $(\matr{W}_j)_{i^\prime i}$ is only non-zero if $i^\prime-i$ is within some small range, for all $j$.
Then, let us define the operators
\begin{align}
\begin{split}
    A_{ij} &= X^{(a)}_j[\vec{e}_i] 
    Z^{(b)}_{j+l_y}[-\matr{W}_j \vec{e}_i]
    Z^{(b)}_{j+l_y+1}[\matr{F}\matr{W}_j \vec{e}_i]\\
    B_{ij} &= X^{(b)}_{j+l_y}[\vec{e}_i] 
    Z^{(a)}_{j}[-\matr{W}_j^T \vec{e}_i]
    Z^{(a)}_{j-1}[\matr{F}^T\matr{W}_j^T \vec{e}_i]
    \end{split}
\end{align}
Notice that these are local operators, as $\matr{W}_j$ is consistently local.
Consider the Hamiltonian
\begin{equation}
    H = -\sum_{ij} A_{ij} - \sum_{ij} B_{ij}.
\end{equation}
which we will now show is symmetric, composed of mutually commuting terms, and has a unique ground state which realizes the desired twist phase $\matr{W}_j$.

First, let us show that $A_{ij}$ and $B_{ij}$ commute with all $S(g^{(a,j_0)}[\vec{v}])$ and $S(g^{(b,j_0)}[\vec{v}])$.
Note that $A_{ij}$ commutes with all $a$ type symmetries, and $B_{ij}$ commutes with all $b$ type symmetries trivially.
To see that $A_{ij}$ commutes with $S(g^{(b,j_0)}[\vec{v}])$, note that the phase factor obtained by commuting the symmetry through the $Z^{(b)}_{j+l_y}$ term exactly cancelled out by the phase from the $Z^{(b)}_{j+l_y+1}$ term.
In the same way, it can be shown the $B_{ij}$ commutes with all $S(g^{(a,j_0)}[\vec{v}])$.
Thus, $H$ is symmetric.

Next, we verify that all terms are mutually commuting.
One can verify that $A_{ij}$ and $B_{i^\prime j}$ with the same $j$ commutes, as the phases from commuting each component individually cancels.
For $A_{ij}$ and $B_{i^\prime,j+1}$, one finds that $A_{ij} B_{i^\prime,j+1} = \alpha B_{i^\prime,j+1}A_{ij}$, where
$\alpha=e^{\frac{2\pi i}{p} (\vec{e}_i^T \matr{F}^T\matr{W}_{j+1}^T \vec{e}_{i^\prime} - \vec{e}_i^T \matr{W}_j^T \matr{F}^T \vec{e}_{i^\prime})}=1$
using the fact that $\matr{W}_{j+1}^T = (\matr{F}^{-1})^T \matr{W}_j^T \matr{F}^T$.
All other terms commute trivially.
Therefore, this Hamiltonian is composed of mutually commuting terms.
The set $\{A_{ij}\}\cup\{B_{ij}\}$ may therefore be thought of as generators of a stabilizer group,  
and the ground state is given by the unique state $\ket{\psi}$ that is a simultaneous eigenstate of all operators, $A_{ij}\ket{\psi}=B_{ij}\ket{\psi}=\ket{\psi}$.
Uniqueness of the ground state follows from the fact that all $A_{ij}$ and $B_{ij}$ are all independent operators, which can be seen simply from the fact that $A_{ij}$ is the only operator which contains $X_{ij}^{(a)}$, and $B_{i,j-l_y}$ is the only operator which contains $X_{ij}^{(b)}$ (all other operators act as $Z_{ij}^{(\alpha)}$ or identity on the site $ij$).

Let show that the ground state is uncharged under all symmetries: $S(g^{(\alpha,j_0)}[\vec{v}])\ket{\psi} = \ket{\psi}$.
We do this by showing that every symmetry operation can be written as a product of terms $A_{ij}$ and $B_{ij}$ in the Hamiltonian.
Let us define a vectorial representation for $A_{ij}$ and $B_{ij}$,
\begin{align}
\begin{split}
    A_{j}[\vec{v}] &= 
    \prod_i A_{ij}^{v_i} = 
    X^{(a)}_j[\vec{v}] 
    Z^{(b)}_{j+l_y}[-\matr{W}_j \vec{v}]
    Z^{(b)}_{j+l_y+1}[\matr{F}\matr{W}_j \vec{v}]\\
    B_{j}[\vec{v}] &= 
    \prod_i B_{ij}^{v_i} = 
    X^{(b)}_{j+l_y}[\vec{v}] 
    Z^{(a)}_{j}[-\matr{W}_j^T \vec{v}]
    Z^{(a)}_{j-1}[\matr{F}^T\matr{W}_j^T \vec{v}]
    \end{split}
\end{align}
and note that
\begin{align}
\begin{split}
&\prod_{j} A_j[\matr{F}^{j-j_0}\vec{v}] \\
=&
S(g^{(a,j_0)}[\vec{v}])
\prod_j
Z^{(b)}_{j+l_y}[-\matr{W}_j \matr{F}^{j-j_0}\vec{v}]
Z^{(b)}_{j+l_y+1}[\matr{F}\matr{W}_j\matr{F}^{j-j_0}\vec{v}]\\
=&
S(g^{(a,j_0)}[\vec{v}])\\
&\times
\left[
\prod_j
Z^{(b)}_{j+l_y}[-\matr{W}_j \matr{F}^{j-j_0}\vec{v}]
\right]
\left[
\prod_j
Z^{(b)}_{j+l_y}[\matr{W}_{j}\matr{F}^{j-j_0}\vec{v}]
\right]\\
=&
S(g^{(a,j_0)}[\vec{v}])
\end{split}
\end{align}
where we have again used the evolution equation $\matr{W}_j = \matr{F}\matr{W}_{j-1}\matr{F}^{-1}$.
Similarly, we may show that
\begin{equation}
    \prod_j B_j[(\matr{F}^T)^{j-l_y-j_0} \vec{v}] = S(g^{(b,j_0)}[\vec{v}])
\end{equation}
Thus, all symmetries may be written as a product of $A_{ij}$ and $B_{ij}$, so therefore the ground state $\ket{\psi}$ has eigenvalue $+1$ under all symmetries.

Next, let us measure the twist phases to verify that this model indeed describes the desired phase.
Consider twisting by the symmetry $g^{(a,0)}[\vec{e}_i]$.
Let us conjugate every term in the Hamiltonian crossing the $j=0$ cut by the truncated symmetry operator $S_\geq (g^{(a,0)}[\vec{e}_i])$.
The only terms which are affected by this conjugation are $B_{i^\prime,0}$ for which $\vec{e}_{i}^T \matr{W}_0^T \vec{e}_{i^\prime} = W^{(l_y,0)}_{i^\prime i} \neq 0$, which are transformed as
\begin{equation}
    B_{i^\prime,0}\rightarrow B_{i^\prime,0}^\prime = e^{-\frac{2\pi i}{p} W^{(l_y,0)}_{i^\prime i}}B_{i^\prime,0}
\end{equation}
on the zeroeth row, and $B_{i^\prime,j}^\prime =B_{i^\prime,j}$ on all other $j\neq 0$,
in the twisted Hamiltonian.
Now, we are curious about the charge of a symmetry $g^{(b,l_y)}[\vec{e}_{i^\prime}]$ in the ground state of this twisted Hamiltonian, which acts as
\begin{align}
    S(g^{(b,l_y)}[\vec{e}_{i^\prime}]) &= \prod_j B_j[(\matr{F}^T)^{j} \vec{e}_{i^\prime}]\\
     &= 
     e^{\frac{2\pi i}{p} W^{(l_y,0)}_{i^\prime i}}
     \prod_j B_j^\prime[(\matr{F}^T)^{j} \vec{e}_{i^\prime}]\\
\end{align}
since the symmetry only includes a single $B_{i^\prime,0}$ on the zeroeth row.
Thus, the ground state of the twisted Hamiltonian (which has eigenvalue $1$ under $B_{ij}^\prime$), has picked up a nontrivial charge under the symmetry $S(g^{(b,l_y)}[\vec{e}_{i^\prime}])$, relative to in the untwisted Hamiltonian.
Indeed, this phase is
\begin{equation}
    \Omega(g^{(b,l_y)}_{i^\prime},g^{(a,0)}_{i}) = e^{\frac{2\pi i}{p} W^{(l_y,0)}_{i^\prime i}}
\end{equation}
which is exactly as desired.
Thus, this model indeed realizes the correct projective representation and describes the desired phase of matter.

Note that these models bear resemblance to the cluster state, and can be understood as a $\mathbb{Z}_p$ version of the cluster state on a particular bipartite graph.  
Suppose we have a graph defined by the symmetric $\mathbb{Z}_p$-valued adjacency matrix $Adj({r,r^\prime})\in\mathbb{Z}_p$, where $r,r^\prime$ label two particular sites.  
Then, the Hamiltonian of a generalized cluster state on this graph is given by
\begin{equation}
H_{clus} = \sum_{\vec{r}} U X_r U^\dagger
\end{equation}
where $U=\prod_{r r^\prime} (\mathrm{CZ}_{r r^\prime})^{Adj({r,r^\prime})}$, and
$\mathrm{CZ}_{r r^\prime}$ is a generalized controlled-Z (CZ) operator on the bond connecting sites $r$ and $r^\prime$.
We define the $\mathbb{Z}_p$ generalization of the CZ operator by $\mathrm{CZ}_{r r^\prime}\ket{z_r z_{r^\prime}} = e^{\frac{2\pi i}{p} z_r z_{r^\prime}}\ket{z_r z_{r^\prime}}$, where $\ket{z_r z_{r^\prime}}$ is the eigenstate of $Z_r$ and $Z_{r^\prime}$ with eigenvalues $e^{\frac{2\pi i z_r}{p}}$ and $e^{\frac{2\pi i z_{r^\prime}}{p}}$ respectively, such that $\mathrm{CZ}_{r r^\prime} X_r \mathrm{CZ}_{r r^\prime}^{\dagger} = X_r Z_{r^\prime}$.
Let us label a site by $r=(i,j,\alpha)$, its $xy$-coordinate and its sublattice index $\alpha\in\{a,b\}$.
Then, the graph relevant to this model is given by the adjacency matrix
\begin{align}
\begin{split}
Adj((i,j,a),(i^\prime,j^\prime,b)) &= 
\begin{cases}
(-\matr{W}_j)_{i^\prime i} & \mathrm{if\;} j^\prime = j+l_y\\
(\matr{F} \matr{W}_j)_{i^\prime i} &  \mathrm{if\;} j^\prime = j+l_y+1\\
0 & \mathrm{else}
\end{cases}\\
&=Adj((i^\prime,j^\prime,b),(i,j,a)) 
\end{split}
\label{eq:adj}
\end{align}
Hence, one can think of each site $(i,j,a)$ as being connected to sites $(i^\prime,j+l_y,b)$  by the adjacency matrix given by $-\matr{W}_j$, and also sites $(i^\prime,j+l_y+1,b)$ via $\matr{F}\matr{W}_j$.
Generically, this graph will be complicated and non-planar.

\section{Irreversibility and Pseudosymmetries}\label{sec:pseudosym}
In this section, we discuss fractal symmetries described by a non-reversible linear cellular automaton (for which fractal SPTs do still exist~\cite{Kubica2018-dp,Devakul2018-ru}), or even originally reversible cellular automata that become irreversible when put on different system sizes (e.g. $L_x$ or $L_y$ that are not powers of $p$).

Fractal symmetries are drastically affected by the total system size.  
For example, consider the Sierpinski fractal SPT~\cite{Devakul2018-ru}, which is
 generated by a non-reversible $f(x)=1+x$ with $p=2$, on a lattice of size $L_x=L_y=2^N$.
In this scenario, there are no non-trivial symmetries at all!  
The total symmetry group $G_\mathrm{tot}=\mathbb{Z}_1$ is simply the trivial group.
Yet, we can still define large operators that in the bulk look like symmetries (i.e. they obey the local cellular automaton rules), but violate the rules only within some boundary region.
The total symmetry group being trivial may be thought of as simply an incommensurability effect, whereby the space-time trajectory of the CA cannot form any closed cycles.
Thus, there is still a sense in which this model obeys a symmetry locally.
This effect is exemplified when one notices that, upon opening boundary conditions, there are no obstacles in defining fractal symmetries and non-trivial SPTs.
In this way, it should still be possible to extract what the SPT phase ``would have been'' if the CA rules had been reversible or if the total system sizes had been chosen more appropriately such the total symmetry group had been non-trivial.
To generalize away from the fixed point and to an actual phase, we must formulate what it means for a perturbation to be ``symmetric'' in a system with a potentially trivial total symmetry group.
We will say that such a model is symmetric under a \emph{pseudosymmetry}, as a symmetry of the full system may not even exist.
Thus, a system may respect a pseudosymmetry, and be in a non-trivial pseudosymmetry protected topological phase (pseudo-SPT), despite not having any actual symmetries!

Let us define what we mean when we say that a system is symmetric under a fractal pseudo-symmetry.
Let us work in the case of two fractal symmetries, so $G=\mathbb{Z}_p^{(f,y)}\times \mathbb{Z}_p^{(\bar{f},\bar{y})}$.
As always, we may decompose the Hamiltonian into a sum of local terms
\begin{equation}
    H = \sum_{ij} H_{ij}
\end{equation}
where each $H_{ij}$ has support within some bounded box.
Suppose $H_{ij}$ has support only on sites with $(x,y)$ coordinates $i_0\leq x \leq i_1$ and $j_0\leq y \leq j_1$, where $\ell_x\equiv i_1-i_0$ and $\ell_y\equiv j_1-j_0$ are of order $1$.
Then, we say that $H_{ij}$ is symmetric under the fractal pseudo-symmetry if it commutes with every
\begin{align}
\begin{split}
    \widetilde{S}_i^{(a,j_0,j_1)} &= \prod_{j=j_0}^{j_1} u_j[g^{(a)}; \matr{F}^{j-j_0}\vec{e}_i]
    \,;\,i_0-k_b \ell_y \leq i \leq i_1-k_a \ell_y
    \\
    \widetilde{S}_i^{(b,j_1,j_0)} &= \prod_{j=j_0}^{j_1} u_j[g^{(b)}; (\matr{F}^T)^{j_1-j}\vec{e}_i]
    \,;\,i_0+k_a \ell_y \leq i \leq i_1+k_b \ell_y
    \end{split}
    \label{eq:pseudosym}
\end{align}
which is enough to replicate how any fractal symmetry would act on this $l_x\times l_y$ square, if they existed for the total system.
Notice that these only involve positive powers of $\matr{F}$, as we do not assume an inverse exists.
Thus, even in the extreme case where $G_\mathrm{tot}$ is trivial, a Hamiltonian may still be symmetric under the fractal pseudo-symmetry $G$ in this way.
In the opposite extreme case where $f(x)$ is reversible and $G_\mathrm{tot}=G^{L_x}$ (as was the topic of the rest of this paper), $H_{ij}$ commuting with all pseudo-symmetries is equivalent to it commuting with all the fractal symmetries in $G_\mathrm{tot}$.
Thus, it is natural to expect that pseudo-symmetries may also protect non-trivial SPT phases.

Indeed, notice that one can perform a twist of a pseudo-symmetry.  
Given a cut, $j_0$, we may use the operator $\widetilde{S}^{(a,j_0,j_0+M)}_i$, for some $1\ll M \ll L_x$, in place of the truncated symmetry operator $S_{\geq}(g^{(a,j_0)}_i)$ from Sec~\ref{sec:localphases}.
This can be used to obtain a twisted Hamiltonian as before by conjugating each term 
\begin{equation}
    H_{ij}\rightarrow \widetilde{S}^{(a,j_0,j_0+R)}_i H_{ij}
    (\widetilde{S}^{(a,j_0,j_0+R)}_i)^\dagger
\end{equation}
for some $1\ll R \ll L_x$ if $H_{ij}$ crosses $j_0$.
Each $H_{ij}$ commuting with all their respective pseudo-symmetries (Eq~\ref{eq:pseudosym}) means that the only terms which may no commute with $\tilde{S}_i^{(a,j_0,j_0+R)}$ are those near $(i,j_0)$ and those at the far-away row $j_0+R$ which are not affected by the twisting process.

Measuring the charge of a pseudo-symmetry is a trickier process, since there is no ``symmetry operator'' which we can measure the charge of in the ground state.
Hence, the charge of a pseudo-symmetry is not so well defined.  
However, we may still measure the charge relative to what it would be in the ground state on the untwisted Hamiltonian, $\ket{\psi}$, which turns out to be well defined.
One approach is to again express the twisted process as the action of some local unitary near $(i,j_0)$, $H_\mathrm{twist}\ket{\psi} = U H U^\dagger\ket{\psi}$, where as before $\mathcal{A}(U)$ is contained within some $(2l_x+1)\times(2l_y+1)$ box ($\mathcal{A}(U)$ is defined in Eq~\ref{eq:mcAdef}).
If the support of $U$ were entirely in this box, then we could measure the change in the charge of a $b$ type symmetry by
\begin{equation}
    \Omega^{(j_0+l_y,j_0)}_{i^\prime i} = \bra{\psi} S^\dagger U^\dagger S U\ket{\psi}
\end{equation}
where $S=\widetilde{S}^{(b,j_0+l_y,j_0)}_{i^\prime}$, and $\ket{\psi}$ is the ground state of $H$ (for convenience we have suppressed the dependence of $U$ and $S$ on $i$, $i^\prime$, etc).
However, if the support of $U$ is not confined to this box, this expectation value may not yield a pure phase.  
One solution is to use a family of larger pseudo-symmetry operators which act the same way within $\mathcal{A}(U)$, and take the limit of the sizes going to infinity.
For example, using $S(n)$ instead of $S$ in the above, defined as
\begin{align}
    S(n) \equiv \widetilde{S}_{i^\prime + k_a p^n}^{(b,j_0+l_y+p^n, j_0-p^n)}
\end{align}
which is shown in Figure~\ref{fig:pseudosymcharge}.
For large $n$ and $i^\prime$ within
\begin{equation}
    -l_x+k_a l_y \leq i^\prime-i\leq l_x + k_b l_y,
\end{equation} this operator acts in the same way as $S$ within $\mathcal{A}(U)$, but is also a valid pseudo-symmetry operator elsewhere as well, except on rows $j_0+l_y+p^n$ and $j_0-p^n$ which are far away, as shown in Figure~\ref{fig:pseudosymcharge}.
Then, we may define 
\begin{equation}
    \Omega^{(j_0+l_y,j_0)}_{i^\prime i} \equiv \lim_{n\rightarrow\infty}
    \bra{\psi} S(n)^\dagger U^\dagger S(n) U\ket{\psi}
    \label{eq:susu}
\end{equation}
which, in the large $n$ limit (while keeping $p^n\ll L_x$), is a pure phase.
In the case where $G_\mathrm{tot}=G^{L_x}$, this will coincide with the twist phases
\begin{equation}
    \Omega^{(j_0+l_y,j_0)}_{i^\prime i} = \Omega(g^{(b,j_0+l_y)}_{i^\prime}, g^{(a,j_0)}_{i})
\end{equation}
discussed earlier.

\begin{figure}
    \centering
    \includegraphics[width=0.5\textwidth]{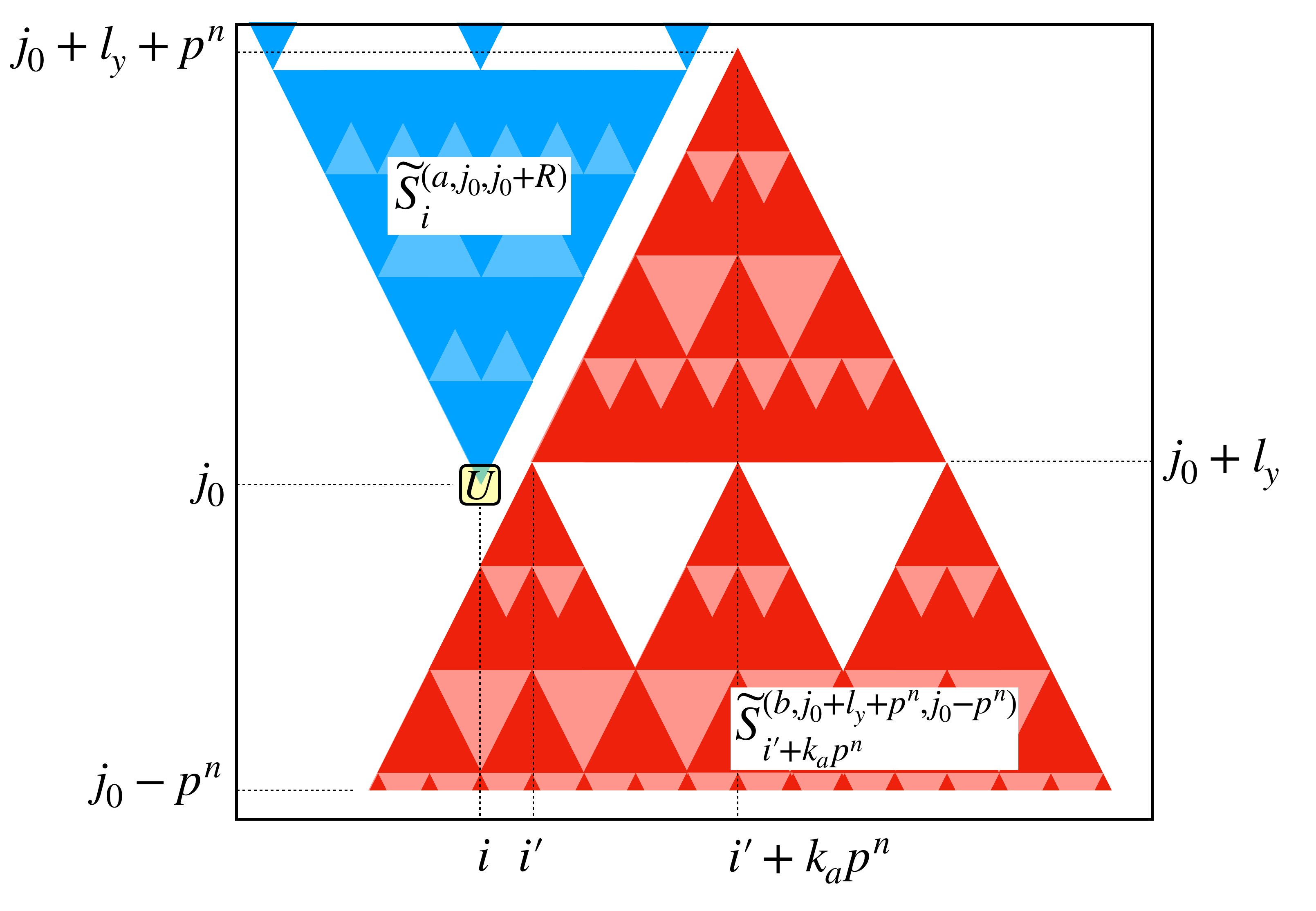}%
    \caption{A visualization of the process to defining a twist phase for pseudosymmetries $\widetilde{S}$.
    Twisting with respect to $\widetilde{S}^{(a,j_0,j_0+R)}_i$ can be thought of as acting via a unitary $U$, which has some region $\mathcal{A}(U)$ shown as the yellow square.  
    To measure the change in charge of another symmetry, we take the expectation value of the commutator (Eq~\ref{eq:susu}) of $U$ with $\widetilde{S}^{(b,j_0+l_y+p^n,j_0-p^n)}_{i^\prime+k_ap^n}$, as shown, for large $n$.
    This may be non-trivial for small $|i^\prime-i|$ when they overlap, and is the sign of a non-trivial pseudo-SPT.
    }
    \label{fig:pseudosymcharge}
\end{figure}

However, the key ingredient to showing that this pseudo-SPT phase is truly stable to local pseudo-symmetric perturbations is to show that $\Omega_{i^\prime i}^{(j_1+l_y,j_1)}$ for all $j_1$ is completely determined by its value at $j_0$.
Define (like before) the matrix $\matr{W}^{(j_1,j_0)}\in\mathbb{F}_p^{L_x\times L_x}$ by
\begin{equation}
    \Omega_{i^\prime i}^{(j_1,j_0)} = e^{\frac{2\pi i}{p} W^{(j_1,j_0)}_{i^\prime i}}
\end{equation}
Starting with $j_0=0$, the matrix $\matr{W}^{(l_y,0)}$ would normally be evolved to $\matr{W}^{(l_y+1,1)}$ using Eq~\ref{eq:Wevolvemat2}.
However, in this case there is no inverse $\matr{F}^{-1}$ which we can use.
Instead, we have the relation
\begin{equation}
    \matr{W}^{(l_y+1,0)} = \matr{W}^{(l_y+1,1)}\matr{F}\label{eq:wwf}
\end{equation}
which does not uniquely determine
$\matr{W}^{(l_y+1,1)}$, as we may add $\vec{v}^T$ to any row of $\matr{W}^{(l_y+1,1)}$, for $\vec{v}\in \ker(\matr{F}^T)$.
However, it is easy to show that any $\vec{v}\in \ker({\matr{F}^T})\backslash\{\vec{0}\}$ is highly non-local, by which we mean that there are no integers $a$ and $b$ for which $v_{i}$ is only non-zero for $a\leq i \leq b$, and $0\leq b-a \ll L_x$ (essentially, any non-zero vector $\vec{v}$ for which $\matr{F}^T\vec{v}=0$ needs to be exploiting the periodic boundary conditions).
Thus, adding any non-zero vector $\vec{v}\in\ker(F^T)$ to a row of a local $\matr{W}^{(l_y+1,1)}$ will necessarily make it non-local.
Thus, \emph{if there exist} a local matrix $\matr{W}^{(l_y+1,1)}$ satisfying Eq~\ref{eq:wwf}, then it is the only local one.
The matrices $\matr{K}^{(k,m)}$ can be defined even for irreversible $f(x)$.
Therefore, for a matrix $\matr{W}^{(l_y,0)}$ composed of a sum of local $\matr{K}^{(k,m)}$, a local $\matr{W}^{(l_y+1,1)}$ does exist and is unique.
This can be reiterated to uniquely determine the set of local $\matr{W}^{(j_0+l_y,j_0)}$ for all $j_0$, assuming it is commensurate with the system size.

Thus, we have shown that $\Omega^{(j_0+l_y,j_0)}_{i^\prime i}$ is indeed a global invariant (knowing it for one $j_0$ determines it for all $j_0$). 
It therefore cannot be changed via a local pseudo-symmetry respecting perturbation (or equivalently a pseudo-symmetry respecting local unitary circuit), and such a phase can indeed be thought of as a non-trivial pseudo-SPT.

To define $\matr{K}^{(k,m)}$ for cases where $f(x)$ may not be reversible, we may simply note that each solution is $p^{N_m}$-periodic in both directions.
Thus, it is straightforward to generalize $\matr{K}^{(k,m)}$ for $m$ where $p^{N_m}$ divides $L_x$ and $L_y$.
In the coming proof, we are careful to show that $f(x)$ is only ever divided out of polynomials of finite degree in $x$ which contained $f(x)$ as a factor anyway, so polynomial division by $f(x)$ is remainder-less and results in another polynomial.
Thus, the results apply equally well for non-reversible $f(x)$, as long as $L_x$ and $L_y$ are commensurate with the periodicity.
This commensurability requirement may greatly reduce the number of possible pseudo-SPT phases, for example, if $L_x$ or $L_y$ are coprime to $p$, then only $m=0$ is allowed.
Note that on such system sizes is also possible to have periodicity that is not a power of $p$ in non-generic cases, for example, the special case where $f(x) = g(x^\nu)$ is a function of only $x^\nu$ and $\nu$ is not a power of $p$.

\section{Identifying the phase}
Suppose we are given an unknown system with $G=\mathbb{Z}_p^{(f,y)}\times\mathbb{Z}_p^{(\bar{f},\bar{y})}$, how do we determine what phase it belongs to and how do we convey compactly what phase it is in?  
Recall that for the case with line-like subsystem symmetries (the topic of Ref~\onlinecite{Devakul2018-ru}), to describe a specific phase required information growing with system size, and so a modified phase equivalence relation was introduced to deal with this.
Such a modified phase equivalence was not needed in this case, and we will show that indeed a specific local phase may be described with a finite amount of information.
Suppose we are given an unknown Hamiltonian $H$.
It is possible to compute the full set of twist phases and construct the $L_x\times L_x$ matrix $\matr{W}^{(0,0)}$.
If the only non-zero matrix elements of $\matr{W}^{(0,0)}$ are within some diagonal band, then we are set.
Otherwise, find the smallest integer $l_y\geq0$ for which $\matr{W}^{(l_y,0)} = \matr{F}^{l_y} \matr{W}^{(0,0)}$ is only non-zero within a diagonal band of width $\ell\sim\mathcal{O}(1)$. 
This is guaranteed to be the case for some $l_y$ (also of $\mathcal{O}(1)$) due to locality.
Note that $\ell$ and $l_y$ are independent of which row we call the zeroeth row.
From the fact that $\matr{W}^{(j_0+l_y,j_0)}$ must also be non-zero only within this diagonal band for all $j_0$, our main result (Theorem~\ref{theorem}) states that $\matr{W}^{(l_y,0)}$ must be a sum
\begin{equation}
    \matr{W}^{(l_y,0)} = \sum_{(k,m)\in loc} C_{km} \matr{K}^{(k,m)}
\end{equation}
where $C_{km}\in\mathbb{F}_p$, and $loc$ is the finite set of all pairs $(k,m)$ where $\matr{K}^{(k,m)}$ is also only non-zero within the same band of width $\ell$.
Thus, this phase is specified fully by our choice of origin, $l_y$, and the finite set of non-zero $C_{km}$.
Furthermore, this description does not depend on $L_x$ and $L_y$, and so it makes sense to say whether two systems of different sizes belong to the same phase.
However, note that unlike with ordinary phases, the choice of origin is important here.  
This procedure may also be done in cases where the symmetry is irreversible, 
the matrix $\matr{W}^{(l_y,0)}$ will instead be defined from the pseudo-symmetry twist phases $\Omega^{(l_y,0)}_{i^\prime,i}$.

\section{Discussion}\label{sec:discussion}
We have therefore asked and answered the question of
what SPT phases can exist protected by fractal symmetries for the type $G=\mathbb{Z}_p^{(f,y)}$, $G=\mathbb{Z}_p^{(f,y)}\times \mathbb{Z}_p^{(\bar{f},\bar{y})}$, or combinations thereof.
If we completely ignore locality along the $x$ direction, effectively compactifying our system into a quasi-1D cylinder with global symmetry group $G_\mathrm{tot}=G^{L_x}$, we would have found that the possible phases are classified by $\mathcal{H}^2[G^{L_x},U(1)]$ which is infinitely large as $L_x\rightarrow\infty$.
What we have shown, however, is that the vast majority of these phases require highly non-local correlations that cannot arise from a local Hamiltonian.
In the case of $G=\mathbb{Z}_p^{(f,y)}$, locality disqualifies all but the trivial phase.
In the $G=\mathbb{Z}_p^{(f,y)}\times\mathbb{Z}_p^{(\bar{f},\bar{y})}$ case, there exists multiple non-trivial phases that are allowed.
If we restrict the twist phases to be local up to some degree, $(l_x,l_y)$, then there are only a finite number of possible phases, independent of total system size $L_x$ and $L_y$.
The number of phases depends only on the combination $\ell\equiv 1+2l_x+4l_y\delta_f$, which is linear in both $l_x$ and $l_y$ (thus demonstrating a kind of holographic principle).
For more general combinations of such fractal symmetries, we have shown that the classification of phases simply amounts to finding pairs of fractal symmetries of the form $(\mathbb{Z}_p^{(f,y)},\mathbb{Z}_p^{(\bar{f},\bar{y})})$ and repeating the analysis above.

Where do other previously discovered 2D systems with fractal symmetries fall into our picture?
The quantum Newman-Moore paramagnet~\cite{Newman1999-fq} is described by the Hamiltonian
\begin{equation}
    H_{NM} = -\sum_{ij} Z_{ij} Z_{i,j+1} Z_{i-1,j-1} - h\sum_{ij} X_{ij}
\end{equation}
$X_{ij}$, $Z_{ij}$, are Pauli matrices acting on the qubit degree of freedom on site $(i,j)$.
The symmetry in our notation is given by $G=\mathbb{Z}_2^{(f,y)}$ with $f(x)=1+x$ (which is irreversible).  
$H_{NM}$ has a phase transition from a spontaneously symmetry-broken phase at $|h|<1$ to a trivial symmetric paramagnet at $|h|>1$.  
Our results would imply that there can be no non-trivial SPT phase in this system.
Thus, all the possibilities in this model are different patterns of broken symmetry.
Next, we have the explicit example of the 2D Sierpinski fractal SPT~\cite{Devakul2018-ru,Devakul2018-di} (which appeared at a gapped boundary in Ref~\onlinecite{Kubica2018-dp}).
This model is isomorphic to the cluster model on the honeycomb lattice, and is described by symmetries $G=\mathbb{Z}_2^{(f,y)}\times\mathbb{Z}_2^{(\bar{f},\bar{y})}$ with $f(x)=1+x$.
With proper choice of unit cell, the Hamiltonian is given by
\begin{align}
\begin{split}
    H_{clus} &= -\sum_{ij} X_{ij}^{(b)} Z_{ij}^{(a)} Z_{i,j-1}^{(a)} Z_{i-1,j-1}^{(a)}\\
    & -\sum_{ij} X_{ij}^{(a)} Z_{ij}^{(b)} Z_{i,j+1}^{(b)} Z_{i+1,j+1}^{(b)}
    \end{split}
\end{align}
Notice that $f(x)=1+x$ with $p=2$ is irreversible for all system sizes, thus these phases should be viewed as pseudo-SPT phases (and indeed every term commutes with all the pseudo-symmetries).
Computing the pseudo-SPT twist phases for $H_{clus}$ gives $\Omega^{(j,j)}_{i^\prime i} = (-1)^{\delta_{i^\prime i}}$.  
Thus, we have simply $\matr{W}^{(0,0)} =\mathbb{1}= \matr{K}^{(0,0)}$.  
A translation invariant model must simply be a sum of $\matr{K}^{(k,0)}$ and this is indeed the case here.
The family of $2$D fractal SPT models described in Ref~\onlinecite{Devakul2018-ru} all realize $\matr{W}^{(0,0)} = \mathbb{1}$.
Our results here imply the existence of a number of new local phases for which the Hamiltonian and twist phases are not strictly translation invariant with period $1$.
Sec~\ref{sec:construct} gave a construction of such models, which works even when $f(x)$ is not reversible.
The twist phases for these models may be translation invariant with a minimal period of $2^n$ sites along either $x$ or $y$, but in exchange will also require interactions of range $\mathcal{O}(2^n)$.

\begin{figure*}
    \centering
    \includegraphics[width=\textwidth]{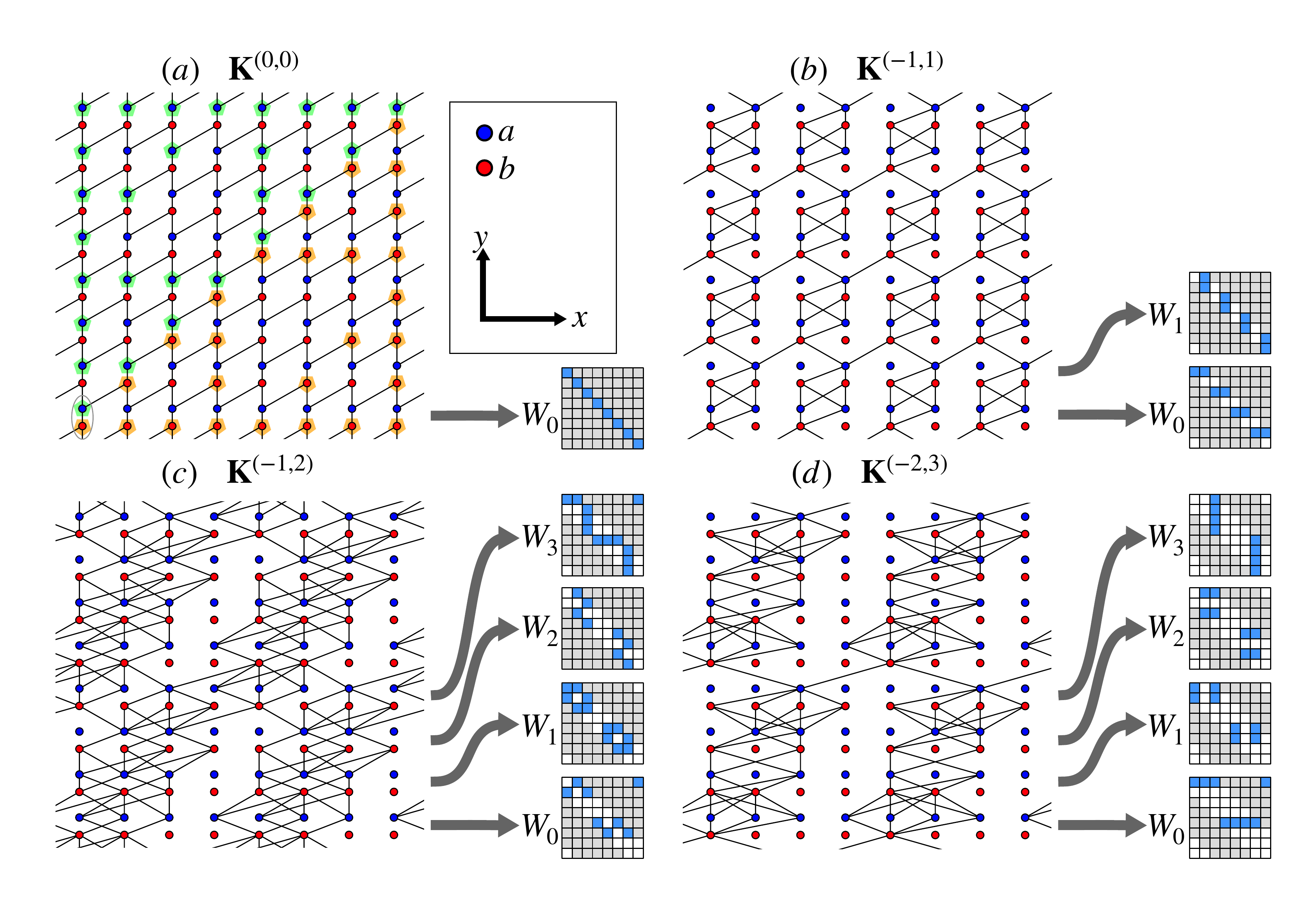}%
    \caption{
    Explicit examples of some possible phases for the case of $G=\mathbb{Z}_2^{(f,y)}\times\mathbb{Z}_2^{(\bar{f},\bar{y})}$ with $f(x)=1+x$: the Sierpinski triangle symmetry.
    The models are constructed following the procedure of Sec~\ref{sec:construct}, and are simply cluster models defined on some underlying graph.
    The models all have $l_y=0$ and $\matr{W}_0=\matr{W}^{(0,0)}$ given by $(a)$ $\matr{K}^{(0,0)}$, $(b)$ $\matr{K}^{(-1,1)}$, $(c)$ $\matr{K}^{(-1,2)}$, and $(d)$ $\matr{K}^{(-2,3)}$.
    Each site (gray circle in $(a)$) consists of an $a$ and a $b$ qubit, which are represented by blue and red vertices.
    Example of fractal subsystems on which the symmetries act are also shown in $(a)$: green highlighted vertices for the $a$ type subsystems, and orange for the $b$.
    The reader may verify that $\prod X$ on these subsystems is indeed a symmetry of the cluster model defined on all these graphs.
    The case $(a)$ is simply the previously studied honeycomb lattice cluster model, which is translation invariant.
    The other three are previously undiscovered phases, and are only translation invariant with a period of $(b)$ $2$ or $(c,d)$ $4$ along the $x$ and $y$ directions.
    The graphs for phases with $\matr{K}^{(k,m)}$ for $k$ other than the ones chosen here may be obtained simply by shifting each blue $a$ vertex left/right by a number of sites, while maintaining connectivity of the graph.
    For each case,
    we also show one cycle of the matrices $\matr{W}_j=\matr{F}^j \matr{W}_0 \matr{F}^{-j}$, the matrix characterizing the twist phases for symmetries defined w.r.t. to row $j$, presented in the same manner as in Fig~\ref{fig:cycles}.  
    The lower-leftmost $a$ and $b$ vertices of each graph are defined to be at coordinate $(x,y)=(0,0)$.
    Although we have only shown examples on an $8\times 8$ torus, these may be tiled onto any commensurate system size.
    }
    \label{fig:otherphases}
\end{figure*}
We show explicitly in Fig~\ref{fig:otherphases} a few of these additional phases that were previously undiscovered, which are represented as cluster models on various graphs.
Recall that the usual $\mathbb{Z}_2$ cluster model for on an arbitrary graph is simply given by the Hamiltonian 
\begin{equation}
H_{clus} = -\sum_{v} X_v\prod_{v^\prime\in \mathrm{adj}(v)} Z_{v^\prime},
\end{equation}
 where the sum is over vertices $v$, and $\mathrm{adj}(v)$ is the set of vertices connected to $v$ by an edge.

A signature of subsystem SPT phases is an extensive protected ground state degeneracy along the edge.
That is, for an edge of length $L_{edge}$, there is a ground state degeneracy scaling as $\log \mathrm{GSD} \sim L_{edge}$ which cannot be lifted without breaking the subsystem symmetries.
The dimension of the protected subspace may be thought of as the minimum dimension needed to realize the projective representation characterizing the phase on the boundary.
For the case of the honeycomb lattice cluster model (Fig~\ref{fig:otherphases}a), we have exactly $\mathrm{GSD}=2^{L_{edge}}$.
For the more complicated models, some of this degeneracy may be lifted, leaving only a fraction $\mathrm{GSD}=2^{\alpha L_{edge}}$ remaining.  
Moreover, the degeneracy along the left or right edges will also generally be different.

\section{Proof of result} \label{sec:proof}

In this section, we will focus on proving the claim in Sec~\ref{subsec:localtwo} that any consistently local matrix $\matr{W}$ must be a linear sum of $\matr{K}^{(k,m)}$, each of which are local.  
We will say that the set of matrices satisfying this property, $\{\matr{K}^{(k,m)}\}$, serve as an \emph{optimal basis} (this term will be precisely defined soon).
Recall that we are dealing with the case where $L_x=p^N$ is a power of $p$ and $L_y$ is chosen such that $f(x)^{L_y}=1$.
We will simply use $L$ to refer to $L_x$ in this section for convenience.

\subsection{Definition and statement}
We will be using the polynomial representation exclusively.  
Let $W(u,v)$ be a Laurent polynomial over $\mathbb{F}_p$ in $u$ and $v$ representing the twist phases, defined according to Eq~\ref{eq:Wpolydef}. 
% Formally, $a_j(x)$ should be viewed as an element of the quotient ring $\mathbb{F}_p[x]/\langle x^{L_x}-1\rangle $, where $\langle x^{L_x}-1\rangle$ is the principle ideal generated by $x^{L_x}-1$
% (note that $x^{L_x}-1$ is generally reducible, so $\mathbb{F}_p[x]/\langle x^{L_x}-1\rangle$ does not form a field).
Formally, periodic boundary conditions $u^{L}=v^{L}=1$ means that we only care about 
the \emph{equivalence class} of $W(u,v)$ in $\mathbb{F}_p[u,v]/\langle u^{L}-1,v^{L}-1\rangle$, 
where $\langle u^{L}-1,v^{L}-1 \rangle$ is the ideal generated by these two polynomials.
Rather than dealing with equivalence classes, we will instead deal with canonical form polynomials:
a polynomial $q(u,v)$ is in \emph{canonical form} if $\du{q(u,v)}<L$ and $\dv{q(u,v)}<L$.
Obviously, canonical form polynomials are in one-to-one correspondence with equivalence classes from $\mathbb{F}_p[u,v]/\langle u^{L}-1,v^{L}-1\rangle$.
Any polynomial with $u$ or $v$-degree larger than $L$ can be brought into canonical form via simply taking $u^a = u^{a\mod L}$ and $v^a = v^{a\mod L}$.
From now on, we will implicitly assume all polynomials have been brought to their canonical form.

% Essentially, this means that for each equivalence class
% $[q(u,v)]\in\mathbb{F}_p[u,v]/\langle u^{L}-1,v^{L}-1\rangle$,
% there exists a canonical polynomial $q_\mathrm{canon}(u,v)\in[q(u,v)]$ satisfying $\du {q_\mathrm{canon}(u,v)} < L$ and $\dv{q_\mathrm{canon}(u,v)}<L$.
% We say that the polynomial is in canonical form if 
% ${q_\mathrm{canon}(u,v)} < L$ and $\dv{q_\mathrm{canon}(u,v)}<L$.
% From now on, all polynomials being in 

Let us now define what it means for a polynomial to be local.
\begin{defn}
A Laurent polynomial $g(u,v)$ is $(a,b)$-local, for integers $a\leq b$, if
\begin{equation}
    \dv v^{-a} g(u,v)  \leq b-a
\end{equation}
\end{defn}
A Laurent polynomial $g(u,v)$ being $(a,b)$-local roughly means that the only non-zero powers of $v$ are $v^a$, $v^{a+1}$, \dots, $v^b$ (powers mod $L$).
As a shorthand, we will more often say that a polynomial is $\ell$-local to mean $(0,\ell)$-local, which can be thought of as simply an upper bound on its $v$ degree.
Whenever something is said to be $\ell$-local, we are usually talking about $\ell\ll L$ being some finite value of order $1$.
Some nice properties are that if $g(u,v)$ is $(a,b)$-local, then 
\begin{enumerate}
    \item   $v^k g(u,v)$ is $(a+k,b+k)$-local;
    \item $g(u,v)g^\prime(u,v)$ is $(a+a^\prime, b+b^\prime)$-local, if $g^\prime(u,v)$ is $(a^\prime,b^\prime)$-local.
    \item $g(u,v)^N$ is $(Na, Nb)$-local;
\end{enumerate}

Next, let us define the ``evolution operator'' $t_h$ with respect to an (invertible) Laurent polynomial $h(x)$ which operates on a polynomial $W(u,v)$ as
\begin{align}
    t_h : W(u,v)\rightarrow t_h(u,v) W(u,v);\\ 
    t_h(u,v) = h(v)^{-1}h(uv)
\end{align}
By invertible, we mean that there exists an inverse $h(v)^{-1}$ with periodic boundary conditions, such that $h(v) h(v)^{-1} = 1$.
In the case of $h(x)=f(x)$, $t_f$ evolves the polynomial $W^{(l_y,0)}(u,v)\rightarrow W^{(l_y+1,1)}(u,v)$.
Notice that an overall shift, $h(x)\rightarrow x^a h(x)$, results in $t_h(u,v)\rightarrow u^{a}t_h(u,v)$, which does not affect the locality properties (which only depends on powers of $v$).
For the purposes of this proof we will therefore simply work with (non-Laurent) polynomials $h(x)$.
We can now define consistent locality.

\begin{defn}
A Laurent polynomial $g(u,v)$ is \emph{consistently $(a,b)$-local under $t_h$} if
$t_h^{n} g(u,v)$ is $(a,b)$-local for all $n$.
\end{defn}

Physically, the twist phases $W^{(l_y,0)}(u,v)$ must be consistently $(-l_x+2l_yk_a,l_x+2l_yk_b)$-local (from Eq~\ref{eq:diagband2}) under $t_f$ in order to correspond to a physical phase obtained from a local Hamiltonian.

Let us define 
\begin{equation}
    U_m(u) = (u-1)^{L-1-m}
\end{equation}
for $m=0\dots L-1$, which serves as a complete basis for all polynomials $g(u)\in\mathbb{F}_p[u]$ with degree $\du g(u) \leq L$.
Any polynomial $W(u,v)$ may therefore be uniquely expanded as 
\begin{equation}
    W(u,v) = \sum_{m=0}^{L-1} U_m(u) W_m(v)
    \label{eq:wuw}
\end{equation}
which we take to be our definition of $W_m(v)$.
Since $U_m(u)$ are all independent, if $W(u,v)$ is $\ell$-local, each $W_m(v)$ must also be $\ell$-local.
\begin{defn}
    A set of polynomials $\{V_m(v)\}$ indexed by  $m=0,\dots,L-1$ is said to \emph{generate an optimal basis} for $t_h$ if for every $\ell$-local $W(u,v)$, $W(u,v)$ is consistently $\ell$-local if and only if $V_m(v)\mid W_m(v)$ for all $m$.
    %, where $W_m(v)$ is defined from $W(u,v)$ in Eq~\ref{eq:wuw}.
    The basis set $\{v^k U_m(u) V_m(v)\}$ is then called an \emph{optimal basis}.
\end{defn}
When we say $V_m(v)\mid W_m(v)$, we mean that $V_m(v)$ divides $W_m(v)$ as polynomials in $\mathbb{F}_p[v]$ without periodic boundary conditions, i.e.
there exists a polynomial $q(v)$ such that 
$W_m(v) = q(v) V_m(v)$ and
\begin{equation}
    \dv q(v) = \dv W_m(v) - \dv V_m(v) \leq \ell-\dv V_m(v)
\end{equation}
which follows by addition of degrees, and since $W_m(v)$ is $\ell$-local.

Suppose $\{V_m(v)\}$ generates an optimal basis for $t_h$.  
Assuming $V_m(v)$ is invertible, $\{v^k U_m(u) V_m(v)\}_{k,m}$ for $0\leq k,m<L$ generates a complete basis for canonical form polynomials.
This basis is optimal with respect to $t_h$ in the sense that all consistently $\ell$-local polynomials under $t_h$ may be written as a linear sum of $\ell$-local basis elements.
If there are only a finite number $N_{loc}$ of $\ell$-local basis elements (as will be the case), then there are also only a finite number $p^{N_{loc}}$ of consistently $\ell$-local $W(u,v)$.

We can now restate our main theorem, 
the proof of which will be the remainder of this section.
\begin{theorem}
The polynomials $\mathcal{V}_m(u,v)$, defined in Eq~\ref{eq:mcKdef}, generate an optimal basis for $t_f$.
\label{theorem}
\end{theorem}
\subsection{Proof}

Let us first list some relevant properties of $U_m(u)$.
\begin{enumerate}
    \item $(u-1)^nU_m(u) = U_{m-n}(u)$ for $n\leq m$, or $0$ for $n>m$\label{uprop1}
    \item $U_m(u)$ is $p^{N_m}$-periodic, meaning
    \begin{equation}
        u^{p^{N_m}} U_m(u) = U_m(u)
    \end{equation}
    where $N_m \equiv \lceil \log_p(m+1)\rceil$.
    This follows from the fact that 
    \begin{equation}
        (u^{p^{N_m}}-1)U_m(u) = 
        (u-1)^{p^{N_m}}U_m(u) = 0
    \end{equation}
    since $p^{N_m}>m$, due to property $1$.
    \label{uprop2}
    \item $U_m(u)$ is also cyclic under evolution by $t_h$ with period dividing $p^{N_m}$, $t_h^{p^{N_m}} U_0(u) = U_0(u)$.
    This follows from the fact that $u^{p^{N_m}} U_0(u) = U_0(u)$, and so 
    \begin{align}
    \begin{split}
    t_h^{p^{N_m}} U_0(u) &=  h(v)^{-p^{N_m}}h(uv)^{p^{N_m}} U_0(u) \\
    &=h(v)^{-p^{N_m}}h(v)^{p^{N_m}} U_0(u) = U_0(u)
    \end{split}
    \end{align}
    \label{uprop3}
    \item $t_h U_m(u) = U_m(u) + \sum_{m^\prime < m} q_{m^\prime}(v) U_{m^\prime}(u)$.
    Under evolution by $t_h$, $t_h U_m(u)$ is given by simply $U_m(u)$, plus a linear combination of $U_{m^\prime}(u)$ for $m^\prime<m$.
    \label{uprop4}
    
\end{enumerate}
Using property~\ref{uprop1},
It is therefore easy to extract each component $W_m(v)$ in the expansion of Eq~\ref{eq:wuw} directly from $W(u,v)$ in a straightforward way.
Suppose the largest $m$ for which $W_m(v)\neq0$ is $m=M$.  
Then, $(u-1)^{M} W(u,v) = W_{M}(v) U_0(u)$ gives only the $m=M$th component multiplying $U_0(u)$.
Then, we may take $W^\prime(u,v) \equiv W(u,v)-U_{M}(u)V_{M}(v)$, which has largest $m$ given by $M^\prime < M$.
This process can be repeated on $W^\prime(u,v)$ to fully obtain $W_m(v)$ for all $m$.
From property~\ref{uprop2}, we also find that $W(u,v)$ is actually $p^{N_m}$-periodic.  

Property $4$ is the most important property (and what makes $U_m(u)$ a special basis for this problem).
It follows from Property~\ref{uprop3} for $m=0$, $t_h U_0(u) = U_0(u)$, and the fact that the $m$th component of $t_h U_m(u)$ is obtained by
\begin{equation}
    (u-1)^m t_h(u,v) U_m(u) =t_h(u,v) U_0(u) = U_0(u)
\end{equation}
which remains unchanged under evolution by $t_h$.
Thus, supposing the expansion of $W(u,v)$ has some largest $m$ value $m=M$, then
defining $\Delta_h W(u,v)$ according to 
\begin{equation}
    t_h W(u,v) = W(u,v) + \Delta_h W(u,v)
\end{equation}
we must have that $\Delta_h W(u,v)$ has a largest $m$ value $m<M$.  
Alternatively, $(u-1)^M \Delta_h W(u,v) = 0$.
This fact will be used numerous times as it allows for a proof by recursion in $M$ in many cases.

Let us first prove two minor Lemmas.
\begin{lemma}
    Suppose $\{V_m(v)\}$ generates an optimal basis for some $t_h$.  
    Then, $V_{m}(v)\mid V_{m^\prime}(v)$ for all $m^\prime \geq m$ and $V_0(v)=1$.
    \label{lemma:optimalgen}
\end{lemma}
\begin{proof}
First, any $\ell$-local $W(u,v)$ that contains only an $m=0$ component,
$W(u,v)=U_0(u)W_0(v)$,
is trivially also consistently $\ell$-local under any $t_h$, since $t_h U_0(u)=U_0(u)$.
Thus, it must be that $V_0(v)=1$.
Next, if $W(u,v)$ is consistently $\ell$-local, then  
\begin{equation}
    (u-1)^n W(u,v) = \sum_{m=n}^{L-1} U_{m-n}(u) W_m(v)
\end{equation}
must also be consistently $\ell$-local for any $n\geq 0$.
However, this implies that $V_m(v)\mid W_{m+n}(v)$.  
But all we know is that $V_{m+n}\mid W_{m+n}(v)$.
For this to always be satisfied, we must therefore have that $V_m(v)\mid V_{m^\prime}(v)$ for all $m^\prime \geq m$.
\end{proof}

\begin{lemma}
Let $W(u,v)$ be $\ell$-local.  Then, $W(u,v)$ is consistently $\ell$-local under $t_h$ if and only if $\Delta_h W(u,v)$ is also consistently $\ell$-local.
\label{lemma:deltalocal}
\end{lemma}
\begin{proof}
Consider evolving $W(u,v)$,
\begin{align}
t_h W(u,v) &= W + \Delta_h W\\
t_h^2 W(u,v) &= W + \Delta_h W + t_h \Delta W\\
t_h^3 W(u,v) &= W + \Delta_h W + t_h \Delta_h W + t_h^2 \Delta_h W
\end{align}
and so on.  
By definition, if $W(u,v)$ is consistently $\ell$-local, then $t_h^n W(u,v)$ must all be $\ell$-local.  
But then, this means that each term added in increasing $n$, $t_h^{n-1} \Delta_h W(u,v)$, must also be $\ell$-local, meaning that
 $\Delta_h W(u,v)$ is therefore consistently $\ell$-local.  
If $W(u,v)$ is not consistently $\ell$-local, then that means that there must be some $n$ such that $t_h^n\Delta W(u,v)$ is not $\ell$-local, which therefore implies that $\Delta_h W(u,v)$ is also not consistently $\ell$-local.
\end{proof}

The next Lemma gives a family of a consistently $\ell$-local polynomials.
\begin{lemma}
    Let $K^h_m(u,v)=U_m(u) h(v)^m$ for some $0\leq m<L$.
    Then, $\dv t_h^n W(u,v) = m\delta_h$ for all $n$.
    It is therefore consistently $m\delta_h$-local under $t_h$.
    \label{lemma:kloc}
\end{lemma}
\begin{proof}
Let us prove by recursion in $m$.  
The base case, $m=0$, is true since $U_0(u)$ is indeed consistently $0$-local.
Now, assume $m>0$ and we have proved this Lemma for all $m^\prime < m$.  

Let us compute $\Delta_h K^h_m(u,v)$,
\begin{align}
    \Delta_h K^h_m(u,v) &= U_m(u) h(v)^{m-1}(h(uv)-h(v)) \\
     &= U_m(u) h(v)^{m-1} \sum_{k=0}^{\delta_h} h_k v^k(u^{\overline{k}}-1)  
\end{align}
where $h(x) = \sum_{x=0}^{\delta_h} h_k x^k$, and we have used Property~\ref{uprop2} of $U_m(u)$ to replace $u^k \rightarrow u^{\overline{k}}$, where $\overline{k} \equiv (k\mod p^{N_m})$ is positive.
Then, we may use the identity
\begin{equation}
    u^{\overline{k}}-1 = \sum_{n=0}^{\overline{k}}\binom{\overline{k}}{n} (u-1)^n
\end{equation}
to expand
\begin{align}
    \Delta_h K^h_m(u,v) &= U_m(u) h(v)^{m-1} \sum_{k=0}^{\delta_h}\sum_{n=0}^{\overline{k}} \binom{\overline{k}}{n} h_k v^k(u-1)^n\\
    &=  \sum_{k=0}^{\delta_h}\sum_{n=0}^{\overline{k}} h_k v^k h(v)^{n-1}K^h_{m-n}(v)
\end{align}
and note that by our recursion assumption,
$v^k h(v)^{n-1} K^h_{m-n}$ is consistently $(k,(m-1)\delta_h+k)$-local.
Since $0\leq k \leq \delta_h$, each term is therefore consistently $m\delta_h$-local.
Thus, $\Delta_h K^h_m(u,v)$ is consistently $m\delta_h$-local.
By Lemma~\ref{lemma:deltalocal}, $K^h_m(u,v)$ is therefore also $m\delta_h$-local.
Finally, the $v$-degree of $K^h_m(u,v)$ saturates $m\delta_h$ since the $m$th component of $t_h^n K^h_m(u,v)$ has $v$-degree $m\delta_h$ for all $n$.
The proof follows for all $m$ by recursion.
\end{proof}

Thus, a family of consistently $\ell$-local $W(u,v)$ may be created by a linear sum over of $\ell$-local $v^kK^h_m(u,v)$, over $k$ and $m$.  
However, this may not be exhaustive: there may be some consistently $\ell$-local $W(u,v)$ that are not in this family.  
To show exhaustiveness, we need to show that $\{V_m(v) = h(v)^m\}$ generates an optimal basis for $t_h$.
This is not true generally, but is true in the case that $h(x)$ is irreducible, which our next lemma addresses.
Notice that $V_m(v) = h(v)^m$ are consistent with the properties of being generators of an optimal basis from Lemma~\ref{lemma:optimalgen}, $V_0(v)=1$ and $V_m(v)\mid V_{m^\prime}(v)$ for all $m^\prime\geq m$.

\begin{lemma}
    Suppose $h(x)$ is an irreducible polynomial. 
    Then,  $\{V_m(v) = h(v)^m\}$ generates an optimal basis for $t_h$.  
    \label{lemma:irreducible}
\end{lemma}
\begin{proof}
To prove that $\{h(v)^m\}$ generates an optimal basis for $t_h$, we need to show that for any $\ell$-local $W(u,v)$, it is consistently $\ell$-local if and only if $h(v)^m\mid W_m(v)$ must hold for all $m$.  

First, the reverse implication follows from Lemma~\ref{lemma:kloc}: if $W(u,v)$ is $\ell$-local and each $h(v)^m \mid W_m(v)$, then $W(u,v)$ is also consistently $\ell$-local.
We must now prove the forward implication.

Let $W(u,v)$ by consistently $\ell$-local under $t_h$, with the expansion
\begin{equation}
    W(u,v) = \sum_{m=0}^M U_m(u) W_m(v)
\end{equation}
where $M$ is the largest $m$ value in the expansion, and $W_M(v)\neq0$.
We need to prove that this implies that $h(v)^m\mid W_m(v)$ for all $m$.
We now prove by recursion, and assume that this has been proven for all $M^\prime < M$.
Note that for the base case $M=0$, $\{h(v)^m\}$ indeed generates an optimal basis for all $M=0$ polynomials $W(u,v)$.
If $h(v) = c v^{k}$ is a monomial, then this proof is also trivial, so we will assume this is not the case in the following.

Consider  $\Delta_h W(u,v)$, which by Lemma~\ref{lemma:deltalocal}, also has maximum $m<M$ and is consistently $\ell$-local.
Take the $m=M-1$th component of $\Delta_h W(u,v)$, obtained by $(u-1)^{M-1} \Delta_h W(u,v)$, which by a straightforward calculation is given by
\begin{equation}
    (u-1)^{M-1} \Delta_h W(u,v) = g(v)h(v)^{-1}W_M(v) U_0(u)
    \label{eq:deltacomp}
\end{equation}
where 
\begin{equation}
     g(v) = \sum_{k=0}^{\delta_h} \overline{k} h_k v^k
\end{equation}
$\delta_h\equiv \dv h(v)$, $h_k$ is defined through $h(v) = \sum_k h_k v^k$, and $\overline{k} \equiv (k\mod p^{N_M})$.
Note that since $W(u,v)$ is $\ell$-local, despite Eq~\ref{eq:deltacomp} containing $h(v)^{-1}$, is of $v$-degree bounded by $\ell$. 
By our recursion assumption, $h(v)^{M-1}$ must divide Eq~\ref{eq:deltacomp}.

Let us prove that $h(v) \nmid g(v)$ and $g(v) \neq 0$. 
First, since $\dv h(v) = \dv g(v)$ and $h(v)$ is irreducible, if $h(v)$ is to possible divide $g(v)$, it must be that $g(v) = \mathrm{const}\cdot h(v)$.  
This can only be the case if $(k\mod p^{N_M})\equiv k_0$ is the same for all $k$.
But then,
\begin{align}
    h(v) &= k_0 \sum_{i=0}^{i_\mathrm{max}-1} h_{k_0+ip^{N_M}}v^{k_0+i p^{N_M}}\\
    &= k_0 v^{k_0}\left( \sum_{i=0}^{i_\mathrm{max}-1} h_{k_0+ip^{N_M}} v^i \right)^{p^{N_M}}
\end{align}
which contradicts with the fact that $h(v)$ is irreducible, as $i_\mathrm{max}>1$ and $p^{N_M}>1$ (which is the case here).  
The $g(v)=0$ is the $k_0=0$ case of this.
Thus, $g(v)\neq0$ and $h(v)\nmid g(v)$.

Going back, we have that
\begin{align}
    h(v)^{M-1} \mid g(v) h(v)^{-1} W_M(v)\\
    \implies h(v)^M \mid g(v) W_M(v)
\end{align}
but since $h(v)\nmid g(v)$, it must be the case that $h(v)^M \mid W_M(v)$.

Now, consider $W^{\prime}(u,v) = W(u,v) - U_M(u)W_M(v)$, which is a sum of two consistently $\ell$-local polynomials (using Lemma~\ref{lemma:kloc}), and so is also consistently $\ell$-local.
By our recursion assumption, it then follows that $h(v)^m \mid W_m(v)$ for $m<M$. 
Thus, $h(v)^m \mid W_m(v)$ holds for all $m$ in $W(u,v)$.

By recursion in $M$, we have therefore proved that for all $W(u,v)$, $h(v)^m \mid W_m(v)$ must be true for all $m$.
Thus, $\{h(v)^m\}$ generates an optimal basis for $t_h$.
\end{proof}

If $\tilde{f}(x)=x^{-k_a} f(x)$ is irreducible, then Lemma~\ref{lemma:irreducible} is sufficient to obtain all consistently $(a,b)$-local $W(u,v)$.
To do so, we simply have to find all $(k,m)$ where the basis element $v^k K^{\tilde{f}}_m(u,v)$ is $(a,b)$-local, and take a linear combination of them.
If there are $N_{loc}(a,b)$ such basis elements, then the $p^{N_{loc}(a,b)}$ possible linear combinations are exhaustive.

In the case that $\tilde{f}(x)$ is not irreducible, there may be consistently $(a,b)$-local polynomials that do not fall within this family.  
However, note that it is always possible to expand $\tilde{f}(x)$ in terms of its unique irreducible factors
\begin{equation}
    \tilde{f}(x) = f_0(x)^{r_0} f_1(x)^{r_1} \dots
\end{equation}
The next two Lemmas allows us to use this result construct an optimal basis for $\tilde{f}(x)$, based on this factorization.

\begin{lemma}
    Let $h(x)$ be an irreducible polynomial, and $r>0$ an integer.  
    Then, $\{V_m(v) = h(v)^{\overline{m}}\}$ 
    generates an optimal basis for $t_{h^r}$,
    where $\overline{m} = \lfloor m/p^\alpha\rfloor p^\alpha$ and $\alpha$ is the power of $p$ in the prime factorization of $r$.
    \label{lemma:power}
\end{lemma}
\begin{proof}
First, note that if $p\nmid r$, $p$ is coprime to $r$, then being consistently $\ell$-local under $t_h$ is equivalent to being consistently $\ell$-local under $t_{h^r}$.
This follows from the fact that, if $W(u,v)$ has maximum $m$ value $m=M$, then $t_h^{p^{N_M}} W(u,v) = W(u,v)$.
If $W(u,v)$ is consistently $\ell$-local under $t_h$, then $t_h^n W(u,v)=t_h^{n\mod p^{N_m}} W(u,v)$ is, by definition, $\ell$-local for all $n$.
If $W(u,v)$ is instead consistently $\ell$-local under $t_{h^r} = t_h^r$, then $t_h^{rn} W(u,v) = t_h^{rn\mod p^{N_m}} W(u,v)$ is $\ell$-local for all $n$.
But, $rn$ takes on all value mod $p^{N_m}$, and so these two conditions are equivalent.
Thus, Lemma~\ref{lemma:irreducible} states that $\{h(v)^m\}$ generates an optimal basis for $t_h$, which therefore also generates an optimal basis for $t_{h^r}$.
Indeed, if $p\nmid r$, $h(v)^{\overline{m}} = h(v)^m$ and the proof is complete.

Next, consider the case where $r=p^\alpha$ is a power of $p$.
Notice that $t_h^r(u,v) = t_h(u^r,v^r)$ in this case is a function of only $u^r$ and $v^r$.
Let $W(u,v)$ be $\ell$-local and decompose it as
\begin{equation}
    W(u,v) = \sum_{i=0}^{r-1}\sum_{j=0}^{r-1} (u-1)^{r-1-i} v^j W_{ij}(u^r,v^r)
\end{equation}
such that each of the $ij$ ``block'' does not mix under evolution by $t_h^r$.  
Thus, each $ij$ may be treated as an independent system in terms of variables $\tilde{u}\equiv u^r$ and $\tilde{v}\equiv v^r$, with $\tilde{L}\equiv L/r$.
Thus, by Lemma~\ref{lemma:irreducible}, each $ij$ component (and therefore $W(u,v)$) is only consistently $\ell$-local if and only if in the decomposition
\begin{equation}
     W_{ij}(u^r, v^r) =  \sum_{\tilde{m}=0}^{L/r-1} (u^r-1)^{L/r-1-\tilde{m}} W_{ij,\tilde{m}}(v^r)
\end{equation}
 $h(v^r)^{\tilde{m}} \mid W_{ij,\tilde{m}}(v^r)$ for all $i,j,\tilde{m}$.
Defining $m\equiv i+\tilde{m}r $, $W(u,v)$ may be written as
\begin{align}
    W(u,v) &= \sum_{m=0}^{L-1}
    U_m(u) \sum_{j=0}^{r-1} v^j W_{ij,\tilde{m}}(u^r,v^r)\\
    &= \sum_{m=0}^{L-1}
    U_m(u) W_m(v)
\end{align}
where $W_m(v) = \sum_{j=0}^{r-1} v^j W_{ij,\tilde{m}}(u^r, v^r)$,
 so $W(u,v)$ is consistently $\ell$-local if and only if $h(v^r)^{\tilde{m}}\mid W_m(v)$.
To eliminate reference to $\tilde{m}$, we may use the fact that $\tilde{m}=\lfloor m / r\rfloor$, such that $\overline{m}= r\tilde{m}$.
Therefore, $W(u,v)$ is consistently $\ell$-local if and only if $h(v)^{\overline{m}}\mid W_m(v)$ for all $m$, and
$\{h(v)^{\overline{m}}\}$ generates an optimal basis for $t_{h^r}$ when $r=p^\alpha$ as well.

Finally, consider the general case $r=\tilde{r}p^\alpha$, where $p\nmid \tilde{r}$.  
We have just shown that $\{h(v)^{\overline{m}}\}$ generates an optimal basis for $t_{h^{p^\alpha}}$.  
Since $\tilde{r}$ is coprime to $p$, by our first argument, this also generates an optimal basis for $t_{h^{r}}$.
\end{proof}
\begin{lemma}
Suppose $\{V_{1,m}(v)\}$ and $V_{2,m}(v)$ generate optimal bases for $t_{h_1}$ and $t_{h_2}$ respectively, and $V_{1,m}(v)$ and $V_{2,m^\prime}(v)$ share no common factors for all $m$, $m^\prime$.
Then, $\{V_m(v) = V_{1,m}(v) V_{2,m}(v)\}$ generates an optimal basis for $t_{h_1 h_2}$.
\label{lemma:combine}
\end{lemma}
\begin{proof}
Let $W(u,v)$ be $\ell$-local which we expand as
\begin{equation}
    W(u,v) = \sum_{m=0}^M U_m(u)W_m(v)
\end{equation}
where $M$ is the largest $m$ for which $W_m(v)\neq 0$.
If $V_{1,m}(v) V_{2,m}(v) \mid W_m(v)$, then $W(u,v)$ is consistently $\ell$-local under $t_{h_1}$ and $t_{h_2}$, and therefore also under $t_{h_1 h_2}$.
We then need to prove the reverse implication,
that $W(u,v)$ being consistently $\ell$-local under $t_{h_1 h_2}$ implies $V_{1,m}(v) V_{2,m}(v) \mid W_m(v)$ for all $m$.
We will prove this by recursion in $M$.
The base case, $M=0$, is trivial since $V_{0}(v)=V_{1,0}(v)V_{2,0}(v)=1$ is a requirement from Lemma~\ref{lemma:optimalgen}.
Now, suppose this has been proven for all $M^\prime<M$.

First, assume that $W(u,v)$ is consistently $\ell$-local under $t_{h_1}$ but not $t_{h_2}$.  
Then, consider $\Delta_{h_1 h_2}W(u,v)$, which has largest $m<M$ and is consistently $\ell$-local under $t_{h_1 h_2}$ by Lemma~\ref{lemma:deltalocal}.  
Our recursion assumption, then, implies that $\Delta_{h_1 h_2}W(u,v)$ is also consistently $\ell$-local under $t_{h_1}$ and $t_{h_2}$ individually. 
Then,
\begin{equation}
    t_{h_1 h_2}^{m} W(u,v) = W(u,v) + \sum_{i=0}^{m-1}t_{h_1 h_2}^i \Delta_{12} W(u,v)
\end{equation}
and so
\begin{equation}
    t_{h_1}^n t_{h_1 h_2}^{m} W(u,v) = t_{h_1}^n W(u,v) + \sum_{i=0}^{m-1}t_{h_1}^nt_{h_1 h_2}^i \Delta_{12} W(u,v)
\end{equation}
which is $\ell$-local.
But, if we choose $n=(k-m \mod p^{N_M})$, then we get that $t_{h_2}^k W(u,v)$ is always $\ell$-local.
Thus, $W(u,v)$ is consistently $\ell$-local under $t_{h_2}$ as well, which contradicts our initial assumption.
Therefore, $W(u,v)$ cannot be consistently $\ell$-local under $t_{h_1}$ but not $t_{h_2}$.  
The same is also true with $h_1\leftrightarrow h_2$.

Next, assume $W(u,v)$ is neither consistently $\ell$-local under $t_{h_1}$ nor $t_{h_2}$.
Then, consider 
\begin{equation}
W^\prime(u,v) \equiv V_{1,M}(v)W(u,v)
\end{equation}
which is consistently $\ell+\dv V_{1,M}(v)\equiv \ell^\prime$-local under $t_{h_1 h_2}$ (notice that if $\ell\ll L$, then $\ell^\prime \ll L$ as well).
$W^\prime(u,v)$ is also $\ell^\prime$-local under $t_{h_1}$, since $V_{1,m}(v) \mid V_{1,M}(v)$ for all $m\leq M$ by Lemma~\ref{lemma:optimalgen}.
However, since $V_{1,M}(v)$ shares no common factors with any $V_{2,m}(v)$, $W(u,v)$ is still not consistently $\ell^\prime$-local under $t_{h_2}$.
But, we just showed previously that we cannot have a situation in which $W(u,v)$ is $\ell^\prime$-local under $t_{h_1h_2}$ and $t_{h_1}$ but not $t_{h_2}$, thus leading to a contradiction.
$W(u,v)$ must therefore be consistently $\ell$-local under both $t_{h_1}$ and $t_{h_2}$.

This means that $V_{1,m}(v)\mid W_m(u,v)$ and $V_{2,m}(v)\mid W_m(u,v)$ for all $m$.  
Since $V_{1,m}(v)$ and $V_{2,m}(v)$ share no common factors, this means that $V_{1,m}(v)V_{2,m}(v)\mid W_m(u,v)$.
Thus, $\{V_{1,m}(v) V_{2,m}(v)\}$ generates an optimal basis for $t_{h_1 h_2}$.
\end{proof}

We may now prove Theorem~\ref{theorem}.
Let us factorize $\tilde{f}(x)$ into its $N_f$ unique irreducible polynomials,
\begin{equation}
    \tilde{f}(x) = \prod_{i=0}^{N_f} f_i(x)^{r_i}
\end{equation}
Using Lemma~\ref{lemma:power}, an optimal basis for $t_{f_i^{r_i}}$, is generated by $\{f_i(v)^{\overline{m}_i}\}$, where $\overline{m}_i = \lfloor m/p^{\alpha_i}\rfloor p^{\alpha_i}$, and $\alpha_i$ is the power of $p$ in the prime factorization of $r_i$. 
Since $f_i(v)^{\overline{m}_i}$ for different $i$ share no common factors (as $f_i(v)$ are irreducible),  Lemma~\ref{lemma:combine} then says that $\{f_0(v)^{\overline{m}_0} f_1(v)^{\overline{m}_1}\}$ generates an optimal basis for $t_{f_{0}^{r_0}  f_{1}^{r_1}}$.
This may be iterated to construct an optimal basis for $t_{f_0^{r_0} f_1^{r_1} f_2^{r_2}}$ and so on.
Finally, one gets that $\{\prod_i f_i(v)^{\overline{m}_i}\}$ generates an optimal basis for $t_{\tilde{f}}$, which is therefore also an optimal basis for $t_f$.  
This is exactly $\{\mathcal{V}_m(v)\}$, and the proof is complete.

\begin{acknowledgements}
T.D. thanks Dominic Williamson, Fiona Burnell, and Beni Yoshida for helpful discussions and comments.
This research was supported in part by the National Science Foundation under Grant No. NSF PHY-1748958 and by the Heising-Simons Foundation.
\end{acknowledgements}

\bibliography{biblio}

\end{document}